\documentclass[12pt]{article}

\usepackage[right=1in,left=1in,top=1in,bottom=1in]{geometry}
\usepackage{amsmath,amssymb,amsthm,mathtools,bm}
\usepackage{booktabs,longtable,array,enumitem}
\usepackage{graphicx}
\usepackage{placeins}
\usepackage{url}
\usepackage[round]{natbib}
\usepackage{xcolor}
\usepackage{microtype}
\usepackage{setspace}
\usepackage{sectsty}
\usepackage{hyperref}
\hypersetup{
  colorlinks,
  citecolor=blue,
  filecolor=blue,
  linkcolor=blue,
  urlcolor=blue
}

\newcommand{\E}{\mathbb E}
\newcommand{\Prb}{\mathbb P}
\newcommand{\Var}{\operatorname{Var}}
\newcommand{\Cov}{\operatorname{Cov}}

\newcommand{\ind}{\mathbf 1}
\newcommand{\dd}{\,\mathrm d}
\newcommand{\Q}{Q}

\newcommand{\norm}[1]{\lVert #1\rVert}

\newtheorem{theorem}{Theorem}
\newtheorem{proposition}{Proposition}
\newtheorem{lemma}{Lemma}
\newtheorem{corollary}{Corollary}
\theoremstyle{definition}
\newtheorem{assumption}{Assumption}

\theoremstyle{remark}
\newtheorem{remark}{Remark}

\renewcommand{\arraystretch}{1.5}

\sectionfont{\large}
\subsectionfont{\normalsize}
\subsubsectionfont{\normalsize}

\title{\vspace*{-2.5cm}Two Margins in Difference-in-Differences with a Continuous Treatment}

\author{
Fangzhou Yu\thanks{School of Economics, University of Sydney. \href{mailto:fangzhou.yu@sydney.edu.au}{fangzhou.yu@sydney.edu.au}}
}

\date{September 2026}

\begin{document}

\hypersetup{pageanchor=false}

\bgroup
\let\footnoterule\relax
\begin{doublespace}
  \maketitle
\end{doublespace}

\begin{singlespace}
  \begin{abstract}
	This paper studies difference-in-differences with staggered adoption and a continuous, time-invariant dose. Each cohort-time comparison contains two margins. The level margin is the average treatment effect at realized doses. Under level parallel trends it equals the level contrast between the treated cohort and not-yet-treated controls. The response margin is the within-cohort slope of the outcome change on dose. It uses no controls, and its causal interpretation requires a response parallel trends assumption and a restriction on selection on gains. We show that the continuous-dose OLS coefficient in each cohort-time comparison is a convex combination of the response index and the level contrast per unit of mean dose, with a mixing weight that depends on the not-yet-treated share. We provide estimators of both margins, joint inference across cohort-time comparisons and event-time aggregates, and a covariate-adjusted extension. In an application to hydraulic fracturing, the level leads reject a joint zero restriction, whereas the response-index leads do not. The continuous-dose OLS coefficient draws primarily on the level margin. We report the level and response margin separately.
  \end{abstract}
\end{singlespace}
\thispagestyle{empty}

\clearpage
\egroup
\setcounter{page}{1}
\hypersetup{pageanchor=true}

\section{Introduction}\label{sec:intro}

Staggered difference-in-differences (DiD) with a binary treatment forms each cohort-time effect from a valid control group and aggregates only afterward \citep{dechaisemartin-dhaultfoeuille-2020,goodmanbacon-2021,sun-abraham-2021,callaway-santanna-2021,borusyak-jaravel-spiess-2024}. With a continuous dose, each cohort-time comparison also contains variation in dose within the treated cohort. A regression of the outcome change on dose uses both the treated-control difference and this within-cohort variation. This paper separates the two comparisons and develops their identification, estimators, and joint inference for the cohort-time and event-time parameters.

The level margin is the cohort's average treatment effect at its realized doses relative to zero treatment. Its observed-data counterpart, the level contrast, is the difference in mean outcome changes between the cohort and units not yet treated. The response margin concerns the effect of increasing dose within the treated cohort. We study this margin through the response index, defined as the population OLS slope of the outcome change on dose within the cohort. The level contrast is measured in outcome units, and the response index in outcome units per unit of dose. Neither determines the other.

We show that the OLS coefficient in each cohort-time comparison is a convex combination of the response index and the level contrast divided by the mean treated dose. Its mixing weight depends on the treated dose distribution and the not-yet-treated share. Even with the conditional outcome and dose distributions held fixed, changing this share changes the coefficient unless the two components coincide.

The two comparisons require separate identifying restrictions. Level parallel trends (PT-L) equates mean untreated outcome changes for the treated cohort and its not-yet-treated controls and identifies the level margin. Response parallel trends (PT-R) removes the relation between dose and untreated outcome changes within the treated cohort. These restrictions are non-nested. PT-R does not rule out selection on gains: under the conditions for the derivative representation, the response index is the sum of a convex weighted average of causal derivatives and a selection term. A bound on the selection derivative gives a sensitivity bound for the causal component.

We present an analysis of county employment using the hydraulic fracturing application of \citet{bartik-et-al-2019}. Under county-level i.i.d. inference, the pre-treatment level leads reject a joint zero restriction, whereas the response-index leads do not. Since the level leads reject and nonrejection of the response pre-trend restriction implies neither PT-R nor the absence of selection on gains, we report both paths descriptively. The continuous-dose OLS coefficient in each cohort-time comparison places most of its weight on the level-contrast component. Removing counties whose missing prospectivity scores were coded as zero changes the level estimates and leaves the response estimates unchanged.

The paper builds on \citet{callaway-goodmanbacon-santanna-2026}, henceforth CGBS. They identify treatment-on-the-treated effects at realized doses under dose-specific parallel trends, show how comparisons across dose groups combine causal responses with selection on gains, and decompose continuous-treatment two-way fixed-effects (TWFE) regressions. Their staggered-adoption extension and event-study aggregations are developed in \citet{callaway-goodmanbacon-santanna-2024-event}. Our level margin is the average of their realized-dose effect over the cohort's dose distribution. In the two-period setting, CGBS's dose-specific parallel-trends condition combines PT-L and PT-R, while their aggregate condition coincides with PT-L.

The baseline estimators are the treated-control mean contrast and the within-cohort OLS slope. Because their influence functions are evaluated on the same units, the stacked covariance retains dependence across cohort-time comparisons and margins, and between cell estimates and estimated cohort-share weights. This yields joint Wald tests and simultaneous inference for fixed collections of comparisons and aggregates. The covariate-adjusted extension preserves the level margin, while residualizing dose within the cohort generally defines a different response index. Building on semiparametric DiD and orthogonal-score methods \citep{abadie-2005,santanna-zhao-2020,chang-2020,chernozhukov-et-al-2018}, we develop cross-fitted estimators and stacked limit theory for the adjusted targets. The response score has the residual-on-residual form used by \citet{clarke-polselli-2026} in their DML estimator for panel data.

Other continuous- or multivalued-treatment DiD methods use different designs or target different objects. \citet{fricke-2017} studies a two-period design with multiple treatments. \citet{dhaultfoeuille-hoderlein-sasaki-2023} use repeated cross-sections and a crossing change in the treatment distribution to identify nonparametric average and quantile effects under stationarity and rank restrictions. For treatments that change over time, \citet{dechaisemartin-et-al-2026-stayers} compare switchers with stayers, \citet{dechaisemartin-dhaultfoeuille-vazquezbare-2024} treat the case in which every unit switches, and \citet{haddad-et-al-2026} develop kernel DML for treatment histories. \citet{zhang-2026} develops DML for kernel-smoothed, dose-specific DiD effects with conditional-density nuisances and uniform-in-dose bands. Both our baseline and adjusted procedures estimate scalar targets without a dose density, a derivative estimator, or a bandwidth.

Section~\ref{sec:setup} defines the staggered design and the two cohort-time targets, and Section~\ref{sec:generic-assumptions} gives their identifying restrictions. Section~\ref{sec:why-two} derives the regression mixture and the causal and selection components of the response index. Section~\ref{sec:paths} defines cohort-specific and event-time parameters. Sections~\ref{sec:baseline-estimation} and~\ref{sec:covariates} develop baseline estimation and inference and the covariate-adjusted extension. Sections~\ref{sec:simulation} and~\ref{sec:application} report simulation results and the hydraulic fracturing application, and Section~\ref{sec:conclusion} concludes. The appendices collect proofs, the covariate-adjusted results and DML derivations, and further application results.

\section{Setup}\label{sec:setup}

This section reduces the staggered-adoption design to cohort-time comparisons, each pairing one treated cohort with units that remain untreated at the outcome date, and defines the two margins within each comparison.

\subsection{Panel structure and cohort-time comparisons}

We consider a balanced panel observed over periods $1,\ldots,\mathcal T$, where $\mathcal T<\infty$.  For one unit, write the observed vector as
\[
  (Y_1,\ldots,Y_{\mathcal T},G,D).
\]
Here $G$ denotes the first treatment date, with $G=\infty$ for never-treated units.  If $G=g<\infty$, the unit receives a scalar dose $D>0$ that remains fixed from period $g$ onward. For never-treated units, set $D=0$. The treatment-path setup follows CGBS.

Let $Y_s(g,d)$ denote the potential outcome at date $s$ under adoption at date $g$ with dose $d$, and let $Y_s(\infty,0)$ denote the never-treated potential outcome. Let $\mathcal C\subseteq\{(g,t):2\le g\le t\le\mathcal T\}$ index the cohort-time comparisons under consideration. For brevity, call each pair $c=(g,t)\in\mathcal C$ a cell. Define
\[
  S_c=\ind\{G=g\text{ or }G>t\},
  \qquad
  T_c=\ind\{G=g\},
  \qquad
  Y_c=Y_t-Y_{g-1},
\]
and let $D_c=T_cD$.  Whenever $P(S_c=1)>0$, let $P_c=P(\,\cdot\mid S_c=1)$ denote the corresponding conditional distribution.  Thus $S_c$ selects cohort-$g$ units and units that remain untreated at date $t$, while $T_c$ distinguishes cohort-$g$ units within this sample.  In particular, already-treated units never serve as controls.

For this cell, define the untreated outcome change and the date-$t$ treatment effect by
\[
  U_c=Y_t(\infty,0)-Y_{g-1}(\infty,0),
  \qquad
  \tau_c(d)=Y_t(g,d)-Y_t(\infty,0).
\]
Consistency and no anticipation imply
\[
  Y_c=U_c+\tau_c(D)\quad\text{when }G=g,
  \qquad
  Y_c=U_c\quad\text{when }G>t.
\]
Conditional on $S_c=1$, each cell is therefore a two-period continuous-treatment DiD problem under $P_c$.

\subsection{Cohort-time target parameters}

The level margin is
\begin{equation}\label{eq:cohort-level-def}
  \tau_{L,g,t}=\E[\tau_c(D)\mid G=g].
\end{equation}
It is the average, within cohort $g$, of each unit's treatment effect at its realized dose relative to zero treatment at outcome date $t$. Whenever $P(G>t)>0$, its observed-data counterpart is the level contrast
\begin{equation}\label{eq:cohort-level-contrast-def}
  \delta_{L,g,t}=\E[Y_c\mid G=g]-\E[Y_c\mid G>t],
\end{equation}
the difference in mean outcome changes between cohort $g$ and the units still untreated at date $t$.  Section~\ref{sec:generic-assumptions} gives the condition under which the two coincide.

The level margin averages the $ATT(g,t,d\mid g,d)$ parameter of \citet{callaway-goodmanbacon-santanna-2026} over the realized-dose distribution $D\mid G=g$, where $g$, $t$, and $d$ index the timing group, outcome period, and dose group in their notation.

The response margin concerns variation across positive doses within cohort $g$.  Let
\[
  \mu_{D,g}=\E[D\mid G=g],
  \qquad
  V_g=\Var(D\mid G=g).
\]
When $V_g>0$, define the response index
\begin{equation}\label{eq:cohort-response-0-def}
  \theta_{R,g,t}
  =
  \frac{\Cov(D,Y_c\mid G=g)}{V_g}
  =
  \E[\alpha_g(D)Y_c\mid G=g],
  \qquad
  \alpha_g(d)=\frac{d-\mu_{D,g}}{V_g}.
\end{equation}
The response index is the coefficient on $D$ in the population linear projection of $Y_c$ on a constant and $D$ within cohort $g$.  Given finite second moments of $D$ and $Y_c$ within cohort $g$, $V_g>0$ is necessary and sufficient for it to be defined. Because the cohort's dose distribution does not vary with the outcome date, $\mu_{D,g}$, $V_g$, $\alpha_g$, and the response weight introduced below carry no $t$ subscript.

Two naming conventions apply throughout the paper.  Objects defined by the observed outcome distribution carry the names "contrast" and "index": the level contrast $\delta_{L,g,t}$ and the response index $\theta_{R,g,t}$ require no causal assumption, and we read the response index in particular as a statistical summary unless the conditions of Section~\ref{sec:why-two} are invoked explicitly.  The name "margin" refers to the causal comparison: the level margin $\tau_{L,g,t}$, and the causal component of the response margin identified in Section~\ref{sec:why-two}.

We report the two margins separately. The level margin is measured in outcome units and the response index in outcome units per unit of dose, so their sum has neither a common scale nor an interpretation as a total effect. The example at the end of the next subsection shows that neither margin determines the other.

\subsection{Generic two-period notation}

To state results that apply to every cohort-time cell, consider a generic two-period problem.  Write
\[
  O=(Y_1,Y_2,D),
  \qquad
  Y=Y_2-Y_1,
  \qquad
  T=\ind\{D>0\},
\]
\[
  U=Y_2(0)-Y_1(0),
  \qquad
  \tau(d)=Y_2(d)-Y_2(0),
\]
so $Y=U+\tau(D)$ under consistency.  Let $p=P(T=1)$, $\rho=P(T=0)=1-p$, and let $\Q=P(\,\cdot\mid T=1)$, with expectation $\E_+[\cdot]$.  Define
\[
  \tau_L=\E_+[\tau(D)],
\]
and define the treated dose moments and, whenever $V_D>0$, the response index and representer by
\[
  \begin{gathered}
    \mu_D=\E_+[D],
    \qquad
    V_D=\Var_+(D),\\
    \theta_R=\frac{\Cov_+(D,Y)}{V_D}
    =\E_+[\alpha(D)Y],
    \qquad
    \alpha(d)=\frac{d-\mu_D}{V_D}.
  \end{gathered}
\]
Under $P_c$, this notation applies cellwise by setting $(Y,T,D)=(Y_c,T_c,D_c)$.  For example, if $D\in\{1,2\}$, a constant effect $\tau(d)=a$ gives $\tau_L=a$ and $\theta_R=0$, whereas a linear effect $\tau(d)=bd$ gives $\tau_L=b\mu_D$ and $\theta_R=b$.

\section{Identification}\label{sec:generic-assumptions}

This section states one parallel-trends restriction for each margin, shows what each identifies, and transfers both to cohort-time cells.

\subsection{Assumptions}

\begin{assumption}[Positivity and moments]\label{ass:sampling}
  (i) $p>0$; (ii) $p<1$; (iii) $\E[Y^2+D^2]<\infty$.
\end{assumption}

\begin{assumption}[Positive treated-dose variation]\label{ass:variation}
  $V_D>0$.
\end{assumption}

Part (i) of Assumption~\ref{ass:sampling} is needed whenever expectations under $\Q$ enter, part (ii) whenever a control group enters, and Assumption~\ref{ass:variation} whenever the response index enters. Under part (i), part (iii) implies $\E_+[Y^2+D^2]<\infty$, so the treated-only projections and covariances below are well defined, and Assumption~\ref{ass:variation} then gives $0<V_D<\infty$.

\begin{assumption}[Consistency, no anticipation, and potential-outcome moments]\label{ass:consistency}
  (i) $Y_1=Y_1(0)$ and $Y_2=Y_2(D)$; (ii) $\E[U^2+\tau(D)^2]<\infty$.
\end{assumption}

Part (i) combines consistency at the outcome date with no anticipation at the pre-treatment date. Part (ii) is a potential-outcome moment condition for the causal covariance decompositions. It is not implied by Assumption~\ref{ass:sampling}(iii), because $U$ and $\tau(D)$ may offset in the observed change. Conversely, under part (i), part (ii) implies the $Y^2$ component of Assumption~\ref{ass:sampling}(iii), and the latter remains useful because it also controls $D^2$ and supports results stated entirely in terms of the observed distribution.

\begin{assumption}[Level parallel trends, PT-L]\label{ass:levelpt}
  \[
    \E[U\mid T=1]=\E[U\mid T=0].
  \]
\end{assumption}

\begin{assumption}[Response parallel trends, PT-R]\label{ass:dosept}
  \[
    \E[U\mid D,T=1]=\E[U\mid T=1]
    \quad\text{almost surely under }\Q.
  \]
\end{assumption}

PT-L is a between-group restriction, whereas PT-R is a within-treated restriction across positive doses.

\begin{proposition}[PT-L and PT-R are non-nested]\label{prop:pt-nonnested}
  Neither Assumption~\ref{ass:levelpt} nor Assumption~\ref{ass:dosept} implies the other.  If both hold, then
  \begin{equation}\label{eq:all-dose-pt}
    \E[U\mid D,T=1]=\E[U\mid T=0]
    \quad\text{almost surely under }\Q.
  \end{equation}
\end{proposition}

In the generic two-period notation, equation~\eqref{eq:all-dose-pt} is the dose-specific parallel-trends condition in \citet{callaway-goodmanbacon-santanna-2026}.  Their condition implies PT-R directly and PT-L after averaging over the positive-dose distribution, and Proposition~\ref{prop:pt-nonnested} gives the converse.  Their aggregate parallel-trends condition coincides with PT-L in this two-period problem. Separating the restrictions therefore does not weaken the parallel-trends assumption but assigns one restriction to each margin.

\subsection{Identification of the two margins}

\begin{theorem}[Identification of the two margins]\label{thm:two-margins}
  \begin{enumerate}[label=(\roman*),leftmargin=*]
    \item Suppose Assumptions~\ref{ass:sampling}(i), (ii) and~\ref{ass:consistency} hold. Under Assumption~\ref{ass:levelpt},
          \[
            \tau_L=\E[Y\mid T=1]-\E[Y\mid T=0].
          \]
    \item Under Assumptions~\ref{ass:sampling}(i), (iii) and~\ref{ass:variation},
          \begin{equation}\label{eq:baseline-balancing}
            \E_+[\alpha(D)]=0,
            \qquad
            \E_+[\alpha(D)D]=1,
          \end{equation}
          and $\theta_R$ is the unique population OLS slope of $Y$ on a constant and $D$ under $\Q$.
    \item Suppose Assumptions~\ref{ass:sampling}(i), (iii), \ref{ass:variation}, and~\ref{ass:consistency} hold. Under Assumption~\ref{ass:dosept},
          \begin{equation}\label{eq:response-gain-id}
            \theta_R
            =
            \frac{\Cov_+\{D,\tau(D)\}}{V_D}.
          \end{equation}
  \end{enumerate}
\end{theorem}

PT-R removes the covariance between dose and untreated trends in part (iii), but selection on gains may remain.  Since
\[
  \theta_R
  =
  \frac{\Cov_+(D,U)}{V_D}
  +
  \frac{\Cov_+\{D,\tau(D)\}}{V_D},
\]
the condition $\Cov_+(D,U)=0$ is necessary and sufficient for \eqref{eq:response-gain-id} when $V_D>0$. We state PT-R rather than this weaker covariance condition because it is the within-treated parallel-trends condition. And together with PT-L, it reproduces the dose-specific condition of CGBS. Section~\ref{sec:why-two} gives the additional conditions under which the response index admits a derivative interpretation.

\subsection{Cellwise assumptions and identification}

For $c=(g,t)\in\mathcal C$, extend the cell treatment effect to zero dose by setting $\widetilde\tau_c(0)=0$ and $\widetilde\tau_c(d)=\tau_c(d)$ for $d>0$.  Under $P_c$, the generic two-period notation applies with
\[
  (Y,T,D,U,\tau(\cdot))
  =
  (Y_c,T_c,D_c,U_c,\widetilde\tau_c(\cdot)).
\]

\begin{assumption}[Panel consistency and no anticipation]\label{ass:staggered}
  Observed outcomes satisfy $Y_s=Y_s(G,D)$ when $G<\infty$ and $Y_s=Y_s(\infty,0)$ when $G=\infty$. For every feasible adoption date $h<\infty$, dose $d>0$, and period $s<h$, $Y_s(h,d)=Y_s(\infty,0)$.
\end{assumption}

Whenever $P(S_c=1)>0$, Assumption~\ref{ass:staggered} gives $Y_c=U_c+\widetilde\tau_c(D_c)$ under $P_c$, the cellwise counterpart of Assumption~\ref{ass:consistency}(i).  The distribution $\Q$ becomes $P(\,\cdot\mid G=g)$, the positivity conditions become $P(G=g)>0$ and $P(G>t)>0$, and $V_D$ becomes $V_g$.  The two margin-specific restrictions become
\[
  \mathrm{PT\text{-}L}_{g,t}:
  \quad
  \E[U_c\mid G=g]=\E[U_c\mid G>t],
\]
and
\[
  \mathrm{PT\text{-}R}_{g,t}:
  \quad
  \E[U_c\mid D,G=g]=\E[U_c\mid G=g]
  \quad\text{almost surely}.
\]

\begin{theorem}[Cohort-time identification]\label{thm:staggered}
  Fix $c=(g,t)\in\mathcal C$ and suppose Assumption~\ref{ass:staggered} holds.
  \begin{enumerate}[label=(\roman*),leftmargin=*]
    \item Suppose $P(G=g)>0$, $P(G>t)>0$,
          \[
            \E[U_c^2+\tau_c(D)^2\mid G=g]<\infty,
            \qquad
            \E[U_c^2\mid G>t]<\infty,
          \]
          and $\mathrm{PT\text{-}L}_{g,t}$ holds. Then
          \[
            \tau_{L,g,t}
            =
            \delta_{L,g,t}
            =
            \E[Y_c\mid G=g]-\E[Y_c\mid G>t].
          \]
    \item Suppose $P(G=g)>0$,
          \[
            \E[D^2+U_c^2+\tau_c(D)^2\mid G=g]<\infty,
          \]
          $V_g>0$, and $\mathrm{PT\text{-}R}_{g,t}$ holds. Then
          \begin{equation}\label{eq:cohort-response-gain-id}
            \theta_{R,g,t}
            =
            \frac{\Cov\{D,\tau_c(D)\mid G=g\}}{V_g}.
          \end{equation}
  \end{enumerate}
\end{theorem}

Figure~\ref{fig:pt-roles} in Appendix~\ref{app:illustrations} illustrates the separate roles of the two restrictions in a stylized cohort-time cell. In Panel~A, PT-L identifies the level margin, while failure of PT-R leaves a covariance between dose and untreated trends in the response index. In Panel~B, PT-R removes that covariance, while failure of PT-L prevents the level contrast from identifying the level margin.

\section{Why separate margins}\label{sec:why-two}

A regression using both zero and positive doses combines the level contrast and the response index in a single coefficient.  This section derives the mixture, decomposes the response index into causal and selection components, and characterizes its representer by balancing and local reweighting.

\subsection{The continuous-dose regression as a mixture}

Let
\[
  \Delta_Y=\E[Y\mid T=1]-\E[Y\mid T=0],
  \qquad
  \beta_{\mathrm{TWFE}}=\frac{\Cov(D,Y)}{\Var(D)}.
\]
The latter is the OLS slope of the long-difference outcome on a constant and $D$.  Equivalently, it is the coefficient on $D_i\ind\{s=2\}$ in the two-period panel regression with unit and time effects.

\begin{theorem}[Two-period TWFE mixture]\label{thm:mixture}
  Suppose Assumptions~\ref{ass:sampling} and~\ref{ass:variation} hold and $D=0$ when $T=0$.  Then
  \begin{equation}\label{eq:share-formula}
    \beta_{\mathrm{TWFE}}
    =
    \frac{V_D\theta_R+\rho\mu_D\Delta_Y}
    {V_D+\rho\mu_D^2}
    =
    \lambda\theta_R
    +(1-\lambda)\frac{\Delta_Y}{\mu_D},
    \qquad
    \lambda=\frac{V_D}{V_D+\rho\mu_D^2}.
  \end{equation}
  Under Assumptions~\ref{ass:consistency} and \ref{ass:levelpt}, $\Delta_Y=\tau_L$, and hence
  \[
    \beta_{\mathrm{TWFE}}
    =
    \lambda\theta_R
    +(1-\lambda)\frac{\tau_L}{\mu_D}.
  \]
  Holding the conditional distributions of $(D,Y)$ given $T=1$ and of $Y$ given $T=0$ fixed while varying $\rho$ gives
  \begin{equation}\label{eq:share-derivative}
    \frac{\partial\beta_{\mathrm{TWFE}}}{\partial\rho}
    =
    \frac{\mu_DV_D(\Delta_Y-\mu_D\theta_R)}
    {(V_D+\rho\mu_D^2)^2}.
  \end{equation}
\end{theorem}

The first component of \eqref{eq:share-formula} is the within-treated response index, and the second is the level contrast $\Delta_Y$ per unit of mean treated dose. Dividing by $\mu_D$ puts both terms in units of outcome per dose.  The mixture weight depends on the dispersion and mean of treated doses and on the untreated share.  Holding the group-specific distributions fixed, an increase in the untreated share shifts weight from the response index toward the level contrast per unit of mean treated dose. As Equation~\eqref{eq:share-derivative} shows, the OLS coefficient can therefore change solely because the relative sizes of the two groups change, even though both components remain fixed. The coefficient is invariant to the untreated share only when $\Delta_Y/\mu_D=\theta_R$.

\begin{corollary}[Cellwise two-period TWFE mixture]\label{cor:share}
  Fix $c=(g,t)$. Suppose $P(G=g)>0$, $P(G>t)>0$,
  \[
    \E[Y_c^2+D_c^2\mid S_c=1]<\infty,
    \qquad
    V_g>0.
  \]
  Let
  \begin{equation}\label{eq:cell-twfe-components}
    \rho_{g,t}=P_c(G>t)
    =\frac{P(G>t)}{P(G=g)+P(G>t)},
  \end{equation}
  and let $\delta_{L,g,t}$ be the level contrast in \eqref{eq:cohort-level-contrast-def}.  The OLS slope of $Y_c$ on a constant and $D_c$ under $P_c$, equivalently the two-period TWFE coefficient for the comparison between cohort $g$ and units not yet treated, is
  \begin{equation}\label{eq:cell-twfe-mixture}
    \beta_{\mathrm{TWFE},g,t}
    =
    \lambda_{g,t}\theta_{R,g,t}
    +(1-\lambda_{g,t})\frac{\delta_{L,g,t}}{\mu_{D,g}},
    \qquad
    \lambda_{g,t}
    =
    \frac{V_g}{V_g+\rho_{g,t}\mu_{D,g}^2}.
  \end{equation}
  Under Assumption~\ref{ass:staggered} and the conditions of Theorem~\ref{thm:staggered}(i), $\delta_{L,g,t}=\tau_{L,g,t}$.
\end{corollary}

Theorem~\ref{thm:mixture} applies to a generic two-period regression, and Corollary~\ref{cor:share} applies it to one cell at a time. As shown for binary treatment by \citet{goodmanbacon-2021}, pooled staggered TWFE can use already-treated cohorts as controls, so changes in their treatment effects enter the comparison. These comparisons are excluded from our cells comparing a treated cohort with units not yet treated, so characterizing the pooled continuous-dose coefficient would require a separate decomposition.

Figure~\ref{fig:two-margins-cell} in Appendix~\ref{app:illustrations} summarizes the two margins in outcome-level and long-difference views. Panel~A illustrates their distinct sources of variation and units. Panel~B visualizes the mixture in Corollary~\ref{cor:share}. The OLS slope using both treated and control units lies between the within-cohort dose slope and the level contrast divided by the mean treated dose.

\subsection{Causal interpretation of the response margin}

Under PT-R, the covariance in Theorem~\ref{thm:two-margins}(iii) may reflect both causal responses and selection on gains. The representation below combines the OLS derivative weight identity of \citet{yitzhaki-1996} with the decomposition in \citet{callaway-goodmanbacon-santanna-2026} of the realized dose effect derivative into a causal response and a selection term.

Define the response weight
\[
  W(d)
  =
  -\E_+[\alpha(D)\ind\{D\le d\}].
\]

\begin{lemma}[Nonnegative response weight]\label{lem:weight-nonnegative}
  Under Assumptions~\ref{ass:sampling}(i), (iii) and~\ref{ass:variation},
  \[
    W(d)\ge0\quad\text{for every }d.
  \]
\end{lemma}

To obtain a derivative representation, impose the following support, smoothness, and compatibility conditions.

\begin{samepage}
\begin{assumption}[Connected support, smoothness, and compatibility]\label{ass:connected}
  (i) The support of $D\mid T=1$ is a compact interval $[a,b]$, $a<b$; (ii) the effect surface
  \[
    \tau(u\mid d)=\E[\tau(u)\mid D=d,T=1]
  \]
  has a version that is continuously differentiable on an open neighborhood of the diagonal $\{(d,d):d\in[a,b]\}$; (iii) the chosen version, with $q(d)=\tau(d\mid d)$, satisfies
  \begin{equation}\label{eq:diagonal-compatibility}
    q(D)=\E[\tau(D)\mid D,T=1]
    \qquad \Q\text{-almost surely}.
  \end{equation}
\end{assumption}
\end{samepage}

Define
\begin{equation}\label{eq:acrt-selection}
  ACRT(d\mid d)
  =
  \left.\partial_u\tau(u\mid d)\right|_{u=d},
  \qquad
  S(d)
  =
  \left.\partial_{d'}\tau(d\mid d')\right|_{d'=d}.
\end{equation}
Here $ACRT(d\mid d)$ is the derivative of the conditional average treatment effect with respect to the assigned dose, evaluated at $u=d$, while $S(d)$ measures how average treatment effects vary across realized-dose groups.  The diagonal $q(d)=\tau(d\mid d)$ of the effect surface is continuously differentiable on $[a,b]$ under Assumption~\ref{ass:connected}, with $q'(d)=ACRT(d\mid d)+S(d)$ by the chain rule.

\begin{theorem}[Response weights and selection]\label{thm:response-decomp}
  Suppose Assumptions~\ref{ass:sampling}(i), (iii), \ref{ass:variation}, \ref{ass:consistency}, \ref{ass:dosept}, and \ref{ass:connected} hold.  Then
  \begin{equation}\label{eq:W-normalized}
    W(d)\ge0,
    \qquad
    \int_a^bW(d)\dd d=1,
  \end{equation}
  \begin{equation}\label{eq:qprime-representation}
    \theta_R=\int_a^bW(d)q'(d)\dd d,
  \end{equation}
  and
  \begin{equation}\label{eq:selection-decomp}
    \theta_R
    =
    \theta_R^{\mathrm{causal}}
    +B_R^{\mathrm{selection}},
  \end{equation}
  where
  \[
    \theta_R^{\mathrm{causal}}
    =
    \int_a^bW(d)ACRT(d\mid d)\dd d,
  \]
  \begin{equation}\label{eq:selection-part}
    B_R^{\mathrm{selection}}
    =
    \int_a^bW(d)S(d)\dd d.
  \end{equation}
  If $S(d)=0$ almost everywhere, then $\theta_R=\theta_R^{\mathrm{causal}}$.  If $S(d)\ge0$ almost everywhere, then $\theta_R^{\mathrm{causal}}\le\theta_R$, with the inequality reversed when $S(d)\le0$.  Finally, if $|S(d)|\le\bar s(d)$,
  \begin{equation}\label{eq:sensitivity-bound}
    |\theta_R-\theta_R^{\mathrm{causal}}|
    \le
    \int_a^bW(d)\bar s(d)\dd d.
  \end{equation}
\end{theorem}

For an absolutely continuous dose, $W$ is the nonnegative OLS derivative weight of \citet{yitzhaki-1996}.  Once PT-R isolates $q$, the pointwise identity $q'(d)=ACRT(d\mid d)+S(d)$ is the decomposition in CGBS, and Theorem~\ref{thm:response-decomp} integrates that identity against $W$.

The restriction $S(d)=0$ is a local no-selection condition that makes the response index a weighted average of causal derivatives without requiring treatment effects to coincide across all realized-dose groups. Equation~\eqref{eq:sensitivity-bound} expresses departures from this condition in outcome units per unit of dose. A useful starting point is a constant envelope, $\bar s(d)=s_0$, which bounds the local variation in average treatment effects across neighboring realized-dose groups, holding the assigned dose fixed. For example, participants choosing more intensive training may have higher returns even when assigned the same training intensity.

Since the response weights integrate to one, the constant envelope implies $\theta_R^{\mathrm{causal}}\in[\theta_R-s_0,\theta_R+s_0]$.  If $[L_R,U_R]$ is an asymptotically valid confidence interval for $\theta_R$, then $[L_R-s_0,U_R+s_0]$ is a conservative asymptotic confidence interval for $\theta_R^{\mathrm{causal}}$ under the selection bound and the assumptions of Theorem~\ref{thm:response-decomp}, and the criterion $L_R-s_0>r_{\min}$ accounts for both sampling uncertainty and the allowed selection when assessing whether the average causal response exceeds a prespecified threshold $r_{\min}$.  Reporting this criterion over several values of $s_0$ shows how much selection on gains the conclusion tolerates, as \citet{rambachan-roth-2023} recommend for violations of parallel trends in binary-treatment DiD, and the additive bound itself parallels the bounded differential trends of \citet{manski-pepper-2018}.  Here PT-R is maintained and $s_0$ bounds selection on gains across realized-dose groups rather than untreated trends, so pre-trend tests do not bound it without an additional assumption linking the two.

\begin{remark}[Relation to CGBS]\label{rem:cgbs-weights}
  \citet[Theorem~3.4(a)]{callaway-goodmanbacon-santanna-2026} decompose $\beta_{\mathrm{TWFE}}$ in Theorem~\ref{thm:mixture} under their dose-specific parallel-trends condition~\eqref{eq:all-dose-pt} and an absolutely continuous treated dose with support $(a,b)$.  With $q'(d)=ACRT(d\mid d)+S(d)$ as in Theorem~\ref{thm:response-decomp}, their decomposition is
  \[
    \begin{gathered}
      \beta_{\mathrm{TWFE}}
      =
      \int_a^b w_1(d)\,q'(d)\dd d
      +w_0\,\frac{q(a)}{a},\\
      w_1(d)=\frac{\E[(D-\E[D])\ind\{D\ge d\}]}{\Var(D)},
      \qquad
      w_0=\frac{(\mu_D-\E[D])\,p\,a}{\Var(D)},
    \end{gathered}
  \]
  where $q(a)$ is the average effect of the lowest dose for the lowest-dose group.  A direct calculation in Appendix~\ref{app:mainproofs} expresses their weights in the notation of Theorem~\ref{thm:mixture}:
  \begin{equation}\label{eq:cgbs-weight-identity}
    w_1(d)
    =
    \lambda W(d)
    +(1-\lambda)\frac{\Q(D>d)}{\mu_D}
    \quad\text{for }d\in(a,b),
    \qquad
    w_0=(1-\lambda)\frac{a}{\mu_D}.
  \end{equation}
  Their causal-response weight at each dose is therefore a mixture of the within-treated response weight $W$ and a level component proportional to the treated-dose survival function.  Under the hypotheses of Theorem~\ref{thm:response-decomp} and PT-L, integrating the second component against $q'$ gives $(1-\lambda)\{\tau_L-q(a)\}/\mu_D$, and the lowest-dose term contributes $(1-\lambda)q(a)/\mu_D$, and their sum is the second term of \eqref{eq:share-formula}.  The level margin thus enters their decomposition through both terms, whereas \eqref{eq:share-formula} collects it in one term and leaves the response term with weights that do not depend on the untreated share.  CGBS observe that $w_1$ depends on the size of the untreated group.  In \eqref{eq:cgbs-weight-identity} that dependence runs only through $\lambda$, and $\rho\to0$ gives $\lambda\to1$ and $\beta_{\mathrm{TWFE}}\to\theta_R$.  Their Remark~3.2 on designs without untreated units is the case $\lambda=1$, in which the coefficient equals $\theta_R$ and their decomposition reduces to Theorem~\ref{thm:response-decomp}, proved there under Assumption~\ref{ass:connected} without a dose density.
\end{remark}

\subsection{Balancing and local reweighting}

The response index is the linear functional $\theta_R=\E_+[\alpha(D)Y]$ of the treated outcome distribution.  This subsection gives two characterizations of the representer $\alpha$ that use neither the derivative representation nor Assumption~\ref{ass:connected}.

The first concerns the two restrictions in \eqref{eq:baseline-balancing}.  Define
\begin{equation}\label{eq:balancing-class-main}
  \mathcal A_R
  =
  \left\{
  a\in L_2(\Q_D):
  \E_+[a(D)]=0,
  \quad
  \E_+[a(D)D]=1
  \right\}.
\end{equation}
For $a\in\mathcal A_R$, the functional $\E_+[a(D)Y]$ is invariant to adding a constant to $Y$ and equals one when $Y=D$.  By Theorem~\ref{thm:two-margins}(ii), $\alpha\in\mathcal A_R$.

\begin{proposition}[Minimum-norm balancing weight]\label{prop:minnorm}
  Under Assumptions~\ref{ass:sampling}(i), (iii) and~\ref{ass:variation}, $\alpha$ is the unique solution to
  \[
    \min_{a\in\mathcal A_R}\E_+[a(D)^2],
  \]
  and
  \[
    \min_{a\in\mathcal A_R}\E_+[a(D)^2]=\frac1{V_D}.
  \]
  If $Y=m(D)+\varepsilon$ under $\Q$, with $\E_+[\varepsilon\mid D]=0$ and $\E_+[\varepsilon^2\mid D]=\sigma^2$ for some $0<\sigma^2<\infty$, the same representer uniquely minimizes $\Var_+\{a(D)\varepsilon\}$ over $\mathcal A_R$.
\end{proposition}

The least-squares estimands of \citet{hines-diazordaz-vansteelandt-2026} impose the same two restrictions. Appendix~\ref{app:response-characterizations} compares the two criteria.

The second characterization reweights the treated-dose distribution.  Let $\mu(d)=\E_+[Y\mid D=d]$.  For $\varepsilon$ near zero, define
\begin{equation}\label{eq:tilt-law}
  \frac{\dd\Q_\varepsilon}{\dd\Q_D}(d)
  =
  \frac{\exp(\varepsilon d)}{M(\varepsilon)},
  \qquad
  M(\varepsilon)=\E_+[\exp(\varepsilon D)],
\end{equation}
and
\[
  \mathcal Y(\varepsilon)=\int\mu(d)\Q_\varepsilon(\dd d),
  \qquad
  \mathcal D(\varepsilon)=\int d\,\Q_\varepsilon(\dd d),
\]
the mean outcome and mean dose under $\Q_\varepsilon$ when the conditional distribution of $Y$ given $D$ is held fixed.  A positive $\varepsilon$ shifts mass toward higher observed doses.

\begin{proposition}[Local tilt interpretation]\label{prop:tilt}
  Suppose Assumptions~\ref{ass:sampling}(i), (iii) and~\ref{ass:variation} hold and $\E_+[\exp(\varepsilon_0|D|)]<\infty$ for some $\varepsilon_0>0$.  Then $\mathcal Y$ and $\mathcal D$ are differentiable at $\varepsilon=0$, with
  \[
    \mathcal Y'(0)=\Cov_+(D,Y),
    \qquad
    \mathcal D'(0)=V_D,
  \]
  so
  \begin{equation}\label{eq:tilt-interpretation}
    \theta_R
    =
    \frac{\mathcal Y'(0)}{\mathcal D'(0)}
    =
    \left.\frac{\dd\mathcal Y}{\dd\mathcal D}\right|_{\varepsilon=0}.
  \end{equation}
\end{proposition}

Along this path, $\theta_R$ is the derivative of the mean outcome with respect to the mean dose at the observed treated-dose distribution.  The path changes the weight on units at each realized dose and does not change any unit's dose.  Appendix~\ref{app:response-characterizations} contains the proofs, the comparison with the stochastic-policy tilts of \citet{jetsupphasuk-et-al-2025}, and the comparison with the global average derivative of \citet{callaway-goodmanbacon-santanna-2026}.

\section{Cohort-specific and event-time parameters}\label{sec:paths}

The preceding representations apply separately to every cohort-time cell. Because a cohort's dose distribution is common across outcome dates, this section states their cohort forms, whose representers and weights are common over $t$, and then defines event-time aggregates.

For cohort $g$, define
\[
  W_g(d)
  =
  -\E[\alpha_g(D)\ind\{D\le d\}\mid G=g].
\]
When a derivative interpretation is invoked, let the support of $D\mid G=g$ be $[a_g,b_g]$, define
\[
  \tau_{g,t}(u\mid d)
  =
  \E[Y_t(g,u)-Y_t(\infty,0)\mid D=d,G=g],
  \qquad
  q_{g,t}(d)=\tau_{g,t}(d\mid d),
\]
and define $ACRT_{g,t}(d\mid d)$ and $S_{g,t}(d)$ as in \eqref{eq:acrt-selection}.

For the cohort tilt, let $\mu_{g,t}(d)=\E[Y_c\mid D=d,G=g]$ and define
\begin{equation*}% \label{eq:cohort-tilt-law-0}
  \frac{\dd P_{g,\varepsilon}}{\dd P(\,\cdot\mid G=g)}(d)
  =
  \frac{\exp(\varepsilon d)}
  {\E[\exp(\varepsilon D)\mid G=g]}.
\end{equation*}
The induced outcome and dose paths are
\begin{equation}\label{eq:cohort-tilt-paths-0}
  \mathcal Y_{g,t}(\varepsilon)
  =
  \int\mu_{g,t}(d)P_{g,\varepsilon}(\dd d),
  \qquad
  \mathcal D_g(\varepsilon)
  =
  \int d\,P_{g,\varepsilon}(\dd d).
\end{equation}

\begin{corollary}[Cohort balancing, tilt, and response weights]\label{cor:cohort-interpretations}
  For a fixed cohort $g$ with $V_g>0$:
  \begin{enumerate}[label=(\roman*),leftmargin=*]
    \item $\alpha_g$ uniquely minimizes
          \[
            \E[\alpha(D)^2\mid G=g]
          \]
          over square-integrable $\alpha$ satisfying $\E[\alpha(D)\mid G=g]=0$ and $\E[\alpha(D)D\mid G=g]=1$, with minimum $1/V_g$.
    \item If $\E[Y_c^2\mid G=g]<\infty$ and $\E[\exp(\varepsilon_0|D|)\mid G=g]<\infty$ for some $\varepsilon_0>0$, then the paths in \eqref{eq:cohort-tilt-paths-0} are differentiable at zero and
          \[
            \theta_{R,g,t}
            =
            \frac{\mathcal Y_{g,t}'(0)}{\mathcal D_g'(0)}.
          \]
    \item $\alpha_g$ and $W_g$ are common to every $t\ge g$, $W_g(d)\ge0$, and, under connected support,
          \[
            \int_{a_g}^{b_g}W_g(d)\dd d=1.
          \]
          If the conditions of Theorem~\ref{thm:staggered}(ii) hold and Assumption~\ref{ass:connected} holds for the cohort distribution $P(\,\cdot\mid G=g)$, with support $[a_g,b_g]$ and effect surface $\tau_{g,t}(\cdot\mid\cdot)$, then
          \[
            \theta_{R,g,t}
            =
            \int_{a_g}^{b_g}W_g(d)
            \{ACRT_{g,t}(d\mid d)+S_{g,t}(d)\}\dd d.
          \]
          The sign and sensitivity conclusions in Theorem~\ref{thm:response-decomp} apply cellwise.
  \end{enumerate}
\end{corollary}

Because $W_g$ depends only on $D\mid G=g$, changes over $t$ in a cohort's response path reflect changes in the relation between outcome and dose under a fixed weighting rule.

For event time $e=t-g$, choose a nonempty cohort set $\mathcal G_e$, with $(g,g+e)\in\mathcal C$ for every included cohort, together with a weighting rule.  Let $w_{g,e}$ denote the resulting population weights, which sum to one over $g\in\mathcal G_e$. Define
\begin{equation*}% \label{eq:event-aggregates}
  \Theta_{L,e}
  =
  \sum_{g\in\mathcal G_e}w_{g,e}\tau_{L,g,g+e},
  \qquad
  \Theta_{R,e}
  =
  \sum_{g\in\mathcal G_e}w_{g,e}\theta_{R,g,g+e}.
\end{equation*}
Only the cohort set and the weighting rule need to be fixed in advance, because the population weights are part of the target, and the numerical weights may be known constants or regular functionals of the treatment-timing distribution.  A common cohort set across event times holds cohort composition fixed, so the aggregates are comparable over $e$.

\section{Estimation and stacked inference}\label{sec:baseline-estimation}

These targets are functions of cohort means, dose moments, and within-cohort covariances, so they admit closed-form estimators.  Stacking their unit-level influence functions then delivers joint inference across margins and cohort-time cells while allowing arbitrary serial dependence within a unit.\footnote{Stacking here refers to the joint treatment of unit-level influence functions across cells and margins for inference.  It is unrelated to the stacked-regression estimator of \citet{cengiz-et-al-2019}, which appends cohort-specific event-study samples and estimates a single regression.}

\subsection{Direct estimators}

For each fixed cell $c=(g,t)$, let $n_g=\sum_i\ind\{G_i=g\}$ and $n_{0,g,t}=\sum_i\ind\{G_i>t\}$.  Define the relevant sample means by
\[
  \overline Y_{c,g}=\frac1{n_g}\sum_{i:G_i=g}Y_{c,i},
  \qquad
  \overline Y_{c,0}=\frac1{n_{0,g,t}}\sum_{i:G_i>t}Y_{c,i},
  \qquad
  \overline D_g=\frac1{n_g}\sum_{i:G_i=g}D_i.
\]
The estimators target the level contrast $\delta_{L,g,t}$ in \eqref{eq:cohort-level-contrast-def} and the response index $\theta_{R,g,t}$ in \eqref{eq:cohort-response-0-def}; under the conditions of Theorem~\ref{thm:staggered}(i) the former equals the level margin.  The estimators are
\[
  \widehat\delta_{L,g,t}
  =
  \overline Y_{c,g}-\overline Y_{c,0},
\]
and
\[
  \widehat\theta_{R,g,t}
  =
  \frac{
    \sum_{i:G_i=g}(D_i-\overline D_g)
    (Y_{c,i}-\overline Y_{c,g})
  }{
    \sum_{i:G_i=g}(D_i-\overline D_g)^2
  }.
\]
The response estimator is the cohort-$g$ OLS slope with an intercept.  Compute
\[
  \widehat V_g
  =
  \frac1{n_g}\sum_{i:G_i=g}(D_i-\overline D_g)^2
\]
once per cohort and reuse it for every $t$.

\subsection{Panel-level influence functions}

Let
\[
  q_g=P(G=g),
  \qquad
  q_{0,g,t}=P(G>t),
\]
\[
  \mu_{1,g,t}=\E[Y_c\mid G=g],
  \qquad
  \mu_{0,g,t}=\E[Y_c\mid G>t].
\]
The corresponding influence functions for the full panel population are
\begin{equation}\label{eq:phiL-0}
  \phi_{L,g,t}^{P}(O)
  =
  \frac{\ind\{G=g\}}{q_g}(Y_c-\mu_{1,g,t})
  -
  \frac{\ind\{G>t\}}{q_{0,g,t}}(Y_c-\mu_{0,g,t}),
\end{equation}
and
\begin{equation}\label{eq:phiR-0}
  \phi_{R,g,t}^{P}(O)
  =
  \frac{\ind\{G=g\}}{q_gV_g}
  \left[
    (D-\mu_{D,g})(Y_c-\mu_{1,g,t})
    -\theta_{R,g,t}(D-\mu_{D,g})^2
    \right].
\end{equation}

\begin{assumption}[Stacked-inference conditions]\label{ass:baseline-inference}
  Full-path observations are i.i.d. across units.  The cell collection $\mathcal C$ is finite, nonrandom, and not selected from the analysis sample.  For every $c=(g,t)\in\mathcal C$, $q_g>0$, $q_{0,g,t}>0$, $V_g>0$,
  \[
    \E[|Y_c|^4+|D|^4\mid G=g]<\infty,
    \qquad
    \E[Y_c^2\mid G>t]<\infty.
  \]
\end{assumption}

Let $\boldsymbol\vartheta$ stack $(\delta_{L,g,t},\theta_{R,g,t})$ in a fixed cell order, and let $\boldsymbol\phi_i^{P}$ stack the corresponding influence function columns on unit $i$.  Plug-in columns replace the probabilities, means, variances, and slopes by their sample analogues and are zero outside the relevant groups.  Every column is evaluated for all $n$ units in the same order, which preserves the cross-cell and cross-margin covariances.

\begin{theorem}[Stacked inference]\label{thm:baseline-stacked}
  Under Assumption~\ref{ass:baseline-inference},
  \begin{equation}\label{eq:baseline-vector-clt}
    \sqrt n
    (\widehat{\boldsymbol\vartheta}
    -\boldsymbol\vartheta)
    =
    \frac1{\sqrt n}\sum_{i=1}^n
    \boldsymbol\phi_i^{P}
    +o_{\Prb}(1)
    \xrightarrow{d}
    N(0,\Sigma),
  \end{equation}
  where $\Sigma=\E[\boldsymbol\phi^{P} \boldsymbol\phi^{P\prime}]$.  Moreover,
  \begin{equation}\label{eq:baseline-sigma}
    \widehat\Sigma
    =
    \mathbb P_n[
    \widehat{\boldsymbol\phi}^{P}_i
    \widehat{\boldsymbol\phi}^{P\prime}_i]
    \xrightarrow{p}\Sigma.
  \end{equation}
\end{theorem}

The covariance matrix of the estimator is estimated by $\widehat\Sigma/n$.  Throughout the paper, Wald and max-$t$ statistics presuppose nonsingular limiting covariances for the coordinates they use.

The influence function stack can be augmented to test a pre-specified finite collection of the cellwise restrictions $\delta_{L,g,t}=\mu_{D,g}\theta_{R,g,t}$ under which the regression coefficient in Corollary~\ref{cor:share} is invariant to the untreated share. The influence function of the cohort dose mean is
\begin{equation*}% \label{eq:phi-mu}
  \phi_{\mu,g}^{P}(O)
  =
  \frac{\ind\{G=g\}}{q_g}(D-\mu_{D,g}).
\end{equation*}
The delta method therefore gives the influence function of
$\delta_{L,g,t}-\mu_{D,g}\theta_{R,g,t}$ as
\begin{equation}\label{eq:phi-proportionality}
  \phi_{L,g,t}^{P}
  -\mu_{D,g}\phi_{R,g,t}^{P}
  -\theta_{R,g,t}\phi_{\mu,g}^{P}.
\end{equation}
Under the conditions of Theorem~\ref{thm:staggered}(i), the same restriction becomes $\tau_{L,g,t}=\mu_{D,g}\theta_{R,g,t}$.

The next result separates prespecification of an aggregate from knowledge of its numerical weights.  It applies to either the stack in Theorem~\ref{thm:baseline-stacked} or any other fixed finite cell stack whose estimators have a joint asymptotic linear representation under the full-panel distribution $P$.  Let $\mathcal J$ index the scalar coordinates of such a stack, with parameters $\vartheta_j$, estimators $\widehat\vartheta_j$, and full-panel influence functions $\phi_j^P$, stacked as $\boldsymbol\phi^P$ in the fixed cell-coordinate order.  Let $r=1,\ldots,K_A$ index a finite family of aggregates, each defined by a coordinate set $\mathcal J_r\subseteq\mathcal J$ and a treatment-timing weighting rule with population weights $w_{rj}$, $\sum_{j\in\mathcal J_r}w_{rj}=1$, and sample weights $\widehat w_{rj}$ with influence columns $\xi^w_{rj}$, stacked as $\boldsymbol\xi^w$ in the fixed aggregate-coordinate order.  We report only within-margin aggregates, so each $\mathcal J_r$ contains coordinates from one margin.  Define
\[
  \eta_r
  =
  \sum_{j\in\mathcal J_r}w_{rj}\vartheta_j,
  \qquad
  \widehat\eta_r
  =
  \sum_{j\in\mathcal J_r}\widehat w_{rj}\widehat\vartheta_j,
\]
and the combined influence function and its plug-in version
\begin{equation}\label{eq:event-psi}
  \psi_{r,i}
  =
  \sum_{j\in\mathcal J_r}
  \left(
  w_{rj}\phi_{j,i}^{P}
  +\vartheta_j\xi^w_{rj,i}
  \right),
  \qquad
  \widehat\psi_{r,i}
  =
  \sum_{j\in\mathcal J_r}
  \left(
  \widehat w_{rj}\widehat\phi_{j,i}^{P}
  +\widehat\vartheta_j\widehat\xi^w_{rj,i}
  \right).
\end{equation}

\begin{theorem}[Inference for weighted aggregates]\label{thm:event-aggregation}
  Let $K_A$ and the finite coordinate sets $\{\mathcal J_r\}_{r=1}^{K_A}$ be fixed as $n\to\infty$, with the coordinate sets and weighting rules nonrandom and not selected from the analysis sample.  Let $Z_i=(\boldsymbol\phi_i^{P\prime},\boldsymbol\xi_i^{w\prime})'$ be a function of the panel observation $O_i$, i.i.d. across $i$, with $\E[Z]=0$ and $\E[\norm{Z}^2]<\infty$.  Suppose that, jointly over $j\in\mathcal J$,
  \[
    \sqrt n(\widehat\vartheta_j-\vartheta_j)
    =
    \frac1{\sqrt n}\sum_{i=1}^n\phi_{j,i}^{P}
    +o_{\Prb}(1),
  \]
  and, jointly over $r=1,\ldots,K_A$ and $j\in\mathcal J_r$,
  \begin{equation}\label{eq:event-weight-alr}
    \sqrt n(\widehat w_{rj}-w_{rj})
    =
    \frac1{\sqrt n}\sum_{i=1}^n\xi^w_{rj,i}
    +o_{\Prb}(1).
  \end{equation}
  Then, for $\boldsymbol\eta=(\eta_1,\ldots,\eta_{K_A})'$ and $\boldsymbol\psi_i=(\psi_{1,i},\ldots,\psi_{K_A,i})'$ from \eqref{eq:event-psi},
  \begin{equation}\label{eq:event-if}
    \sqrt n(\widehat{\boldsymbol\eta}-\boldsymbol\eta)
    =
    \frac1{\sqrt n}\sum_{i=1}^n\boldsymbol\psi_i
    +o_{\Prb}(1)
    \xrightarrow{d}
    N(0,\Omega),
    \qquad
    \Omega=\E[\boldsymbol\psi\boldsymbol\psi'].
  \end{equation}
  If the plug-in columns $\widehat{\boldsymbol\psi}_i$ satisfy $\mathbb P_n[\norm{\widehat{\boldsymbol\psi}_i-\boldsymbol\psi_i}^2]=o_{\Prb}(1)$, then
  \begin{equation}\label{eq:event-sigma}
    \widehat\Omega
    =
    \mathbb P_n[
    \widehat{\boldsymbol\psi}_i
    \widehat{\boldsymbol\psi}_i']
    \xrightarrow{p}\Omega.
  \end{equation}
\end{theorem}

The covariance matrix of $\widehat{\boldsymbol\eta}$ is estimated by $\widehat\Omega/n$.

For the event-time application and margin $m\in\{L,R\}$, take $r=(m,e)$ and $j=(m,g,g+e)$, with $\xi^w_{rj}=\xi^w_{g,e}$.  Thus the same estimated cohort share enters the two separate margin aggregates, and the stacked covariance retains the induced cross-margin dependence.

An event-time aggregate of $\delta_{L,g,t}-\mu_{D,g}\theta_{R,g,t}$ uses the same formula, with the cell influence function in \eqref{eq:phi-proportionality}.  When cohort shares are estimated, its weight term in \eqref{eq:event-psi} multiplies the cell value $\delta_{L,g,t}-\mu_{D,g}\theta_{R,g,t}$.

When the numerical weights are known and imposed in estimation, set $\widehat w_{rj}=w_{rj}$ and $\xi^w_{rj}=0$.  If $L$ is the resulting fixed aggregation matrix, the influence function and covariance formulas in \eqref{eq:event-if} and \eqref{eq:event-sigma} reduce to $\boldsymbol\psi=L\boldsymbol\phi^{P}$ and $\widehat\Omega=L\widehat\Sigma L'$, the formulas for a fixed linear map.

For cohort-share weighting, let $q_g=P(G=g)$, $Q_e=P(G\in\mathcal G_e)$, and suppose $Q_e>0$.  The population and sample weights are
\[
  w_{g,e}=\frac{q_g}{Q_e},
  \qquad
  \widehat w_{g,e}
  =
  \frac{n_g}{\sum_{h\in\mathcal G_e}n_h},
\]
and the weight influence function is
\[
  \xi^w_{g,e}(O)
  =
  \frac{
  \ind\{G=g\}
  -w_{g,e}\ind\{G\in\mathcal G_e\}
  }{Q_e}.
\]
Its plug-in version replaces $(w_{g,e},Q_e)$ by $(\widehat w_{g,e},\widehat Q_e)$, where $\widehat Q_e=n^{-1}\sum_i\ind\{G_i\in\mathcal G_e\}$.  These influence functions have mean zero and sum to zero over $g\in\mathcal G_e$.  The cohort counts used in the weights come from the full panel, not from the units satisfying $S_c=1$.

\subsection{Pre-treatment comparisons}\label{sec:pretrend}

Let $\mathcal C_{\mathrm{pre}}\subseteq\{(g,t):3\le g\le\mathcal T,\ 1\le t<g-1\}$ be finite, nonrandom, and not selected from the analysis sample. For $c=(g,t)\in\mathcal C_{\mathrm{pre}}$, define
\[
  C_c^{\mathrm{pre}}=\ind\{G>g\},
  \qquad
  S_c^{\mathrm{pre}}=\ind\{G=g\}+C_c^{\mathrm{pre}},
  \qquad
  Y_c^{\mathrm{pre}}=Y_t-Y_{g-1}.
\]
The control group excludes cohort $g$ and remains untreated at both dates. The pre-trend targets are the observed-data functionals
\begin{equation*}% \label{eq:pretrend-targets}
  \begin{aligned}
    \delta_{L,g,t}^{\mathrm{pre}}
    &=\E[Y_c^{\mathrm{pre}}\mid G=g]-\E[Y_c^{\mathrm{pre}}\mid G>g],\\
    \theta_{R,g,t}^{\mathrm{pre}}
    &=\frac{\Cov(D,Y_c^{\mathrm{pre}}\mid G=g)}{V_g}.
  \end{aligned}
\end{equation*}
Estimate them by the corresponding sample mean contrast and within-cohort OLS slope.  Their influence functions are \eqref{eq:phiL-0} and \eqref{eq:phiR-0} with $Y_c^{\mathrm{pre}}$ in place of $Y_c$, $G>g$ in place of $G>t$, and the pre-treatment means, probabilities, and response index in place of the post-treatment ones.

\begin{corollary}[Pre-trend inference]\label{cor:pretrend-inference}
  Suppose observations are i.i.d. across units and, for every $c\in\mathcal C_{\mathrm{pre}}$, $P(G=g)>0$, $P(G>g)>0$, $V_g>0$, and
  \[
    \E[|Y_c^{\mathrm{pre}}|^4+|D|^4\mid G=g]<\infty,
    \qquad
    \E[(Y_c^{\mathrm{pre}})^2\mid G>g]<\infty.
  \]
  Then Theorem~\ref{thm:baseline-stacked} holds for the pre-trend stack, and under Assumption~\ref{ass:baseline-inference} also for the pre-trend and post-treatment stacks jointly; Theorem~\ref{thm:event-aggregation} applies to fixed pre-trend aggregates.
\end{corollary}

Corollary~\ref{cor:pretrend-inference} yields a pre-trend test for each margin.  With $\mathcal C_{\mathrm{pre}}$ and the weighting rule fixed in advance, stack that margin's pre-trend aggregates with their influence columns from \eqref{eq:event-psi} and compare the Wald statistic $n\widehat{\boldsymbol\eta}'\widehat\Omega^{-1}\widehat{\boldsymbol\eta}$ with a $\chi^2$ distribution on as many degrees of freedom as there are aggregates. \citet{roth-2022} shows that pre-trend tests in binary-treatment DiD often have low power and that conditioning on passing them can enlarge bias and reduce coverage, so we report both margins regardless of the tests and do not read nonrejection as support for PT-L or PT-R.  \citet{rambachan-roth-2023} instead bound the post-treatment violation by the pre-trends and obtain confidence sets valid under that bound; the jointly normal pre-treatment and post-treatment stacks supply their inputs. For the response index the bound covers dose-related trends, not selection on gains, which \eqref{eq:sensitivity-bound} bounds instead.

\section{Covariate adjustment}\label{sec:covariates}

Covariates matter in two ways: they provide a conditional identification route for the same level margin, while residualizing dose defines a different response index.

\subsection{Adjusted targets and conditional identification}

For cohort $g$, define
\[
  e_g(x)=\E[D\mid G=g,X=x],
  \qquad
  V_g^{(X)}
  =
  \E[(D-e_g(X))^2\mid G=g],
\]
and, for cell $c=(g,t)$,
\[
  g_{g,t}(x)=\E[Y_c\mid G=g,X=x].
\]
When $V_g^{(X)}>0$, the adjusted response index is
\begin{equation}\label{eq:cohort-response-X-def}
  \theta_{R,g,t}^{(X)}
  =
  \frac{
  \E[(D-e_g(X))\{Y_c-g_{g,t}(X)\}\mid G=g]
  }{V_g^{(X)}}.
\end{equation}

Total covariance and total variance give the exact relation
\begin{equation}\label{eq:response-0-X-relation}
  \theta_{R,g,t}
  =
  \frac{
  V_g^{(X)}\theta_{R,g,t}^{(X)}
  +\Cov\{e_g(X),g_{g,t}(X)\mid G=g\}
  }{
  V_g^{(X)}+\Var\{e_g(X)\mid G=g\}
  }.
\end{equation}
Conditioning on $X$ therefore changes the residual dose variation and the index itself, and it may yield $V_g^{(X)}=0$ when $V_g>0$.  When $X$ is constant, \eqref{eq:cohort-response-X-def} reduces to \eqref{eq:cohort-response-0-def}.  Adjusted response objects carry the superscript $(X)$.

The level margin $\tau_{L,g,t}$ is still \eqref{eq:cohort-level-def}.  Let $m_{0,g,t}(x)=\E[Y_c\mid G>t,X=x]$ and define the adjusted level contrast
\begin{equation}\label{eq:cohort-level-id-X}
  \delta_{L,g,t}^{(X)}
  =
  \E[Y_c-m_{0,g,t}(X)\mid G=g].
\end{equation}
The identifying restrictions are conditional versions of those in Section~\ref{sec:generic-assumptions}.

\begin{assumption}[Conditional cohort-time conditions]\label{ass:covariate-identification}
  For every $c=(g,t)\in\mathcal C$: (i) $X$ is a pre-treatment vector common to every outcome date and cell, and $P(G>t\mid X=x)>0$ for almost every $x$ in the support of $X\mid G=g$; (ii) conditional level parallel trends, $\mathrm{PT\text{-}L}^{(X)}_{g,t}$: $\E[U_c\mid G=g,X]=\E[U_c\mid G>t,X]$; (iii) conditional response parallel trends, $\mathrm{PT\text{-}R}^{(X)}_{g,t}$: $\E[U_c\mid D,G=g,X]=\E[U_c\mid G=g,X]$ almost surely; (iv) $V_g^{(X)}>0$.
\end{assumption}

Under Assumptions~\ref{ass:staggered} and~\ref{ass:covariate-identification}(i), (ii), $\delta_{L,g,t}^{(X)}=\tau_{L,g,t}$.  Conditional PT-L does not imply PT-L, so $\delta_{L,g,t}^{(X)}$ and $\delta_{L,g,t}$ may differ.  Conditional PT-L and conditional PT-R are non-nested, by the examples in the proof of Proposition~\ref{prop:pt-nonnested} applied conditionally on $X$.  Part (i) fixes one covariate vector for every outcome date, so $e_g$ and $V_g^{(X)}$ are common over $t$ within a cohort, as $\mu_{D,g}$ and $V_g$ are.  The conditions for a causal derivative reading of the adjusted index under part (iii) are the conditional analogues of those in Section~\ref{sec:why-two}, stated in Appendix~\ref{app:dmlproofs}.

Partialling out functions of $X$ from the cell regression does not isolate the adjusted response index: the dose coefficient retains a level component (Proposition~\ref{prop:adjusted-mixture}, Appendix~\ref{app:dmlproofs}).

\subsection{Cross-fitted DML and stacked inference}

Work under $P_c$, the distribution conditional on eligibility for cell $c$, and let $\pi_{g,t}(x)=P_c(T_c=1\mid X=x)$ denote the cell propensity.  Together with $m_{0,g,t}$, $e_g$, and $g_{g,t}$, it forms the four nuisance functions.  The propensity, control regression, and treated outcome regression vary by cell.  The dose regression $e_g$ is defined conditional on $G=g$, so each fold-specific estimate uses only cohort-$g$ observations and is reused across all outcome dates for that cohort.

The corresponding orthogonal scores are given in equations~\eqref{eq:level-score} and \eqref{eq:response-score} of Appendix~\ref{app:dmlproofs}.  The response score is a treated-cohort analogue of the residual-on-residual score in \citet{clarke-polselli-2026}, and the level score is an augmented DiD score.  One fixed $K$-fold unit partition is reused across every cell and nuisance function, so that the stacked influence vectors are evaluated for the same units in the same order.  Let $k(i)$ denote unit $i$'s fold and form all predictions out of fold.  With $\mathbb P_{n,c}[h]=\mathbb P_n[S_ch]/\mathbb P_n[S_c]$, the resulting cross-fitted ratio estimators are
\begin{equation}\label{eq:deltahat-X}
  \widehat\delta_{L,g,t}^{(X)}
  =
  \frac{
    \mathbb P_{n,c}\left[
      T_c\{Y_c-\widehat m_{0,-k}(X)\}
      -(1-T_c)
      \frac{\widehat\pi_{-k}(X)}{1-\widehat\pi_{-k}(X)}
      \{Y_c-\widehat m_{0,-k}(X)\}
      \right]
  }{\mathbb P_{n,c}[T_c]},
\end{equation}
and
\begin{equation}\label{eq:thetahat-X}
  \widehat\theta_{R,g,t}^{(X)}
  =
  \frac{
    \mathbb P_{n,c}\left[
      T_c\{D-\widehat e_{-k}(X)\}
      \{Y_c-\widehat g_{-k}(X)\}
      \right]
  }{
    \mathbb P_{n,c}\left[
      T_c\{D-\widehat e_{-k}(X)\}^2
      \right]
  }.
\end{equation}

The product rates in the next assumption come from the exact nuisance remainders in Appendix~\ref{app:dmlproofs}.  The asymptotic expansion uses no identifying restriction; conditional PT-L and PT-R govern only interpretation.

\begin{assumption}[Cross-fitted DML conditions]\label{ass:dml}
  Full-path observations are i.i.d. across units.  Outcomes, doses, and fitted nuisance functions are bounded by fixed constants; $p>0$ and $V^{(X)}>0$; and, for some $\epsilon\in(0,1/2)$, $\epsilon\le\pi(X)\le1-\epsilon$.  Estimated propensities are truncated to $[\epsilon/2,1-\epsilon/2]$.  The number of folds is fixed, and every prediction for unit $i$ is made by a model trained without unit $i$'s fold.  Uniformly over folds, the four nuisances are $L_2$ consistent, with norms evaluated under the distribution of $X\mid S_c=1$ for $(m_0,\pi)$ and under the distribution of $X\mid G=g$ for $(e,g)$, and
  \[
    \norm{\widehat m_0-m_0}_2
    \norm{\widehat\pi-\pi}_2
    =o_{\Prb}(n^{-1/2}),
  \]
  \begin{equation}\label{eq:response-products}
    \norm{\widehat e-e}_2
    \norm{\widehat g-g}_2
    +\norm{\widehat e-e}_2^2
    =o_{\Prb}(n^{-1/2}).
  \end{equation}
\end{assumption}

The squared term in \eqref{eq:response-products} requires $\norm{\widehat e-e}_2=o_{\Prb}(n^{-1/4})$.  Appendix~\ref{app:dmlproofs} states the generic-cell asymptotic linear representation and efficiency result, Theorem~\ref{thm:joint-dml}.

\begin{assumption}[Cellwise DML conditions]\label{ass:staggered-dml}
  The collection $\mathcal C$ is finite, nonrandom, and not selected from the analysis sample.  For every $c\in\mathcal C$, $P(S_c=1)>0$ and Assumption~\ref{ass:dml} holds under $P_c$.
\end{assumption}

Let $\phi_{L,c}^{P_c}$ and $\phi_{R,c}^{(X),P_c}$ denote the gradients under $P_c$ derived in equations~\eqref{eq:phiL} and \eqref{eq:phiR} of Appendix~\ref{app:dmlproofs}.  Because $P_c$ varies across cells, each gradient must first be transported to the common full-panel distribution $P$:
\begin{equation}\label{eq:cell-full-if}
  \boldsymbol\phi_c^{(X),P}(O)
  =
  \frac{S_c}{P(S_c=1)}
  \begin{pmatrix}
    \phi_{L,c}^{P_c}(O) \\
    \phi_{R,c}^{(X),P_c}(O)
  \end{pmatrix}.
\end{equation}

\begin{theorem}[Stacked covariate-adjusted DML]\label{thm:staggered-dml}
  Under Assumption~\ref{ass:staggered-dml}, let $\boldsymbol\vartheta^{(X)}$ stack $(\delta_{L,g,t}^{(X)},\theta_{R,g,t}^{(X)})$ over $c\in\mathcal C$ in a fixed order, let $\widehat{\boldsymbol\vartheta}^{(X)}$ stack the estimators \eqref{eq:deltahat-X} and \eqref{eq:thetahat-X}, computed with one unit-level fold partition for every cell and nuisance function, and let $\boldsymbol\phi^{(X),P}$ stack the rescaled columns in \eqref{eq:cell-full-if}, both in the same order.  Then
  \begin{equation}\label{eq:staggered-vector-clt}
    \sqrt n
    (\widehat{\boldsymbol\vartheta}^{(X)}
    -\boldsymbol\vartheta^{(X)})
    =
    \frac1{\sqrt n}\sum_{i=1}^n
    \boldsymbol\phi_i^{(X),P}
    +o_{\Prb}(1)
    \xrightarrow{d}
    N(0,\Sigma^{(X)}),
  \end{equation}
  where $\Sigma^{(X)}=\E[\boldsymbol\phi^{(X),P}\boldsymbol\phi^{(X),P\prime}]$.  Let $\widehat{\boldsymbol\phi}_i^{(X),P}$ stack the plug-in columns $\{S_{c,i}/\mathbb P_n[S_c]\}\widehat{\boldsymbol\phi}_c^{P_c}(O_i)$, with $\widehat{\boldsymbol\phi}_c^{P_c}$ the plug-in gradient of Theorem~\ref{thm:joint-dml}, evaluated for every unit $i$ in the same order.  Then
  \begin{equation}\label{eq:staggered-sigma}
    \widehat\Sigma^{(X)}
    =
    \mathbb P_n[
    \widehat{\boldsymbol\phi}_i^{(X),P}
    \widehat{\boldsymbol\phi}_i^{(X),P\prime}]
    \xrightarrow{p}\Sigma^{(X)}.
  \end{equation}
\end{theorem}

Because \eqref{eq:staggered-vector-clt} is an asymptotic linear representation under $P$, Theorem~\ref{thm:event-aggregation} applies to the adjusted stack.  Cohort-share weights and their influence functions are computed on the full panel without cross-fitting and enter \eqref{eq:event-psi} without the rescaling in \eqref{eq:cell-full-if}.

\section{Simulation}\label{sec:simulation}

We study the finite-sample performance of the baseline level and response estimators in a staggered-adoption design.

We generate a balanced panel over $t=1,\ldots,7$.  Adoption dates satisfy
\[
  \Pr(G_i=4)=\Pr(G_i=5)=\Pr(G_i=6)=0.2,
  \qquad
  \Pr(G_i=\infty)=0.4.
\]
The fixed post-treatment cell collection contains cohorts $g\in\{4,5,6\}$ at event times $e=t-g\in\{0,1\}$, giving six cohort-time cells.  For treated units,
\[
  D_i=\mu+h(2B_i-1),
  \qquad
  B_i\sim\operatorname{Beta}(3,3),
  \qquad
  (\mu,h)=(2,1.8),
\]
so the positive-dose support is $[0.2,3.8]$, and never-treated units have $D_i=0$.  The untreated potential outcome is
\[
  Y_{it}(\infty,0)
  =A_i+m_t+\sum_{r=1}^t\varepsilon_{ir}
  +\mathbf{1}\{G_i<\infty\}\,t\{b_L+b_R(D_i-\mu)\},
  \qquad
  m_t=0.1t+0.02t^2,
\]
where $A_i$ and the innovations $\varepsilon_{ir}$ are independent standard normal variables and are independent of $(G_i,D_i)$. The performance and benchmark designs set $b_L=b_R=0$, so PT-L and PT-R hold. The pre-trend experiment below varies these two trend parameters. For a unit in cohort $g$ with realized dose $d$, the treatment-effect surface at event time $e$ on the positive-dose support is
\begin{equation}\label{eq:simulation-effect-surface}
  \tau_{g,e}(u\mid d)
  =
  3+a_g r_e\left\{0.5u+\kappa h
  \tanh\!\left(3\frac{u-\mu}{h}\right)\right\}
  +s(d-\mu),
\end{equation}
where $(a_4,a_5,a_6)=(0.8,1,1.2)$ and $(r_0,r_1)=(1,1.4)$. The selection derivative in Theorem~\ref{thm:response-decomp} is $s$. The centered $\tanh$ term averages to zero under the symmetric treated-dose distribution, so varying $\kappa$ changes the response index while leaving the level margin unchanged.

Two designs, both with $s=0$, differ only in $\kappa$: the nonlinear design $\kappa=0.4$ is the main design, and the linear design sets $\kappa=0$. Symmetry gives $\mu_D=2$, and the Beta variance gives $V_D=h^2/7$. At each event time, we aggregate the three cohort estimates using sample cohort shares computed from full-panel counts. Their population values are one third. For each margin, the overall aggregate weights event times zero and one equally. In the main design, the aggregate level margin is $4.2$, and the aggregate response index is $1.4436$.

We use $n\in\{500,1{,}000,2{,}000\}$ and 5,000 replications at each sample size. The cell estimators are those of Section~\ref{sec:baseline-estimation}, and aggregate standard errors use the combined influence function in \eqref{eq:event-psi}.

% Set row spacing here so table regeneration preserves the manuscript layout.
\AddToHookNext{env/tabular/before}{\renewcommand{\arraystretch}{1.0}}
\begin{table}[!htbp]
  \centering
  \caption{Finite-sample performance of the two-margin estimators}
  \label{tab:simulation-staggered}
  \small
  \renewcommand{\arraystretch}{1.15}
  \begin{tabular}{rcccc}
    \toprule
    $n$ & Bias & SD & Mean SE & Coverage \\
    \midrule
    \multicolumn{5}{l}{\textit{Panel A: Level margin}} \\
    500 & 0.0001 & 0.102 & 0.101 & 0.9516 \\
    1{,}000 & 0.0016 & 0.072 & 0.072 & 0.9472 \\
    2{,}000 & 0.0001 & 0.051 & 0.051 & 0.9472 \\
    \addlinespace
    \multicolumn{5}{l}{\textit{Panel B: Response index}} \\
    500 & 0.0044 & 0.099 & 0.097 & 0.9406 \\
    1{,}000 & 0.0026 & 0.070 & 0.069 & 0.9480 \\
    2{,}000 & -0.0003 & 0.048 & 0.049 & 0.9514 \\
    \bottomrule
  \end{tabular}
  \begin{minipage}{0.98\textwidth}
    \begin{singlespace}
    \footnotesize
    \textit{Notes:} 5{,}000 replications; $n$ is the number of units. The nonlinear design has no selection. The level and response targets are 4.2 and 1.4436, respectively. Both aggregates use full-panel sample cohort shares and equal weights on $e=0,1$. SD is the empirical standard deviation; Mean SE is the average influence-function standard error, including uncertainty in cohort shares and cross-cell covariance. Coverage refers to 95\% Wald intervals.
    \end{singlespace}
  \end{minipage}
\end{table}

Table~\ref{tab:simulation-staggered} reports results for both margins. Mean standard errors are close to the empirical standard deviations, and coverage of the 95 percent intervals is close to the nominal level. For the response aggregate, ignoring covariance between the two event-time estimates lowers coverage to 0.8568--0.8722.

Table~\ref{tab:simulation-benchmark} compares the within-cohort OLS estimator with the \texttt{contdid} estimator of the average causal response on the treated (ACRT), using the same cohort and event-time weights.\footnote{The benchmark is the ACRT estimator in \texttt{contdid}~0.1.1 with not-yet-treated controls, a varying base period, and the default cubic specification with no interior knots.} In the linear design the causal response is constant in dose within each cell, so the response index and the ACRT coincide at $0.6$ and the within-cohort regression is correctly specified. Both estimators have little bias, and within-cohort OLS has 22.1--25.3 percent lower empirical variance across the three sample sizes, the parsimony gain of a single correctly specified slope over a cubic dose curve. In the nonlinear design the two targets differ. The response index weights the causal response by $W$ and equals $1.4436$, whereas the ACRT averages it over the dose distribution and equals $1.3655$. Against its own target, within-cohort OLS again has smaller absolute bias and standard deviation and coverage closer to 95 percent, while the \texttt{contdid} bias includes an approximation error of $-0.0154$ from its fixed cubic specification.

\AddToHookNext{env/tabular/before}{\renewcommand{\arraystretch}{1.0}}
\begin{table}[!htbp]
  \centering
  \caption{Comparison of response estimators}
  \label{tab:simulation-benchmark}
  \small
  \renewcommand{\arraystretch}{1.15}
  \setlength{\tabcolsep}{4pt}
  \begin{tabular}{rcccccc}
    \toprule
    & \multicolumn{3}{c}{Within-cohort OLS} & \multicolumn{3}{c}{\texttt{contdid}} \\
    \cmidrule(lr){2-4} \cmidrule(lr){5-7}
    $n$ & Bias & SD & Coverage & Bias & SD & Coverage \\
    \midrule
    \multicolumn{7}{l}{\textit{Panel A: Linear design}} \\
    500 & -0.003 & 0.095 & 0.946 & -0.003 & 0.110 & 0.938 \\
    1{,}000 & -0.001 & 0.067 & 0.947 & -0.001 & 0.076 & 0.946 \\
    2{,}000 & 0.001 & 0.047 & 0.951 & 0.001 & 0.053 & 0.947 \\
    \addlinespace
    \multicolumn{7}{l}{\textit{Panel B: Nonlinear design}} \\
    500 & 0.004 & 0.099 & 0.941 & -0.022 & 0.118 & 0.930 \\
    1{,}000 & 0.003 & 0.070 & 0.948 & -0.017 & 0.080 & 0.936 \\
    2{,}000 & 0.000 & 0.048 & 0.951 & -0.018 & 0.055 & 0.937 \\
    \bottomrule
  \end{tabular}
  \begin{minipage}{0.98\textwidth}
    \begin{singlespace}
    \footnotesize
    \textit{Notes:} 5{,}000 replications; $n$ is the number of units. Aggregation uses full-panel sample cohort shares and equal weights on $e=0,1$. Both targets equal 0.6 in Panel A. In Panel B, the response index is 1.4436 and the ACRT is 1.3655; \texttt{contdid} bias includes cubic approximation error. SD is the empirical standard deviation. Coverage is for 95\% Wald intervals.
    \end{singlespace}
  \end{minipage}
\end{table}

\begingroup
\let\simulationtable\table
\renewcommand{\table}[1][]{\simulationtable[!htb]}
\AddToHookNext{env/tabular/before}{\renewcommand{\arraystretch}{1.0}}
\begin{table}[!htbp]
  \centering
  \caption{Empirical rejection rates of pre-trend tests}
  \label{tab:simulation-pretrend-tests}
  \small
  \renewcommand{\arraystretch}{1.15}
  \begin{tabular}{lrcc}
    \toprule
    Design $(b_L,b_R)$ & $n$ & Level test & Response-index test \\
    \midrule
    Null $(0,0)$ & 500 & 0.049 & 0.056 \\
     & 1{,}000 & 0.050 & 0.055 \\
     & 2{,}000 & 0.054 & 0.053 \\
    \addlinespace
    Level trend $(0.10,0)$ & 500 & 0.264 & 0.067 \\
     & 1{,}000 & 0.467 & 0.054 \\
     & 2{,}000 & 0.778 & 0.047 \\
    \addlinespace
    Level trend $(0.20,0)$ & 500 & 0.771 & 0.059 \\
     & 1{,}000 & 0.976 & 0.055 \\
     & 2{,}000 & 1.000 & 0.054 \\
    \addlinespace
    Dose-related trend $(0,0.10)$ & 500 & 0.053 & 0.311 \\
     & 1{,}000 & 0.050 & 0.563 \\
     & 2{,}000 & 0.049 & 0.853 \\
    \addlinespace
    Dose-related trend $(0,0.20)$ & 500 & 0.055 & 0.861 \\
     & 1{,}000 & 0.051 & 0.995 \\
     & 2{,}000 & 0.049 & 1.000 \\
    \bottomrule
  \end{tabular}
  \begin{minipage}{0.98\textwidth}
    \begin{singlespace}
    \footnotesize
    \textit{Notes:} 5{,}000 replications per design and sample size; $n$ is the number of units. Entries are rejection rates of separate joint Wald tests of the level leads and response-index leads at $e=-3,-2$, relative to $e=-1$. Each test uses a 5\% nominal level and two degrees of freedom. Level comparisons use $G>g$ controls. Covariance estimates retain dependence across leads and account for full-panel sample cohort shares.
    \end{singlespace}
  \end{minipage}
\end{table}

\endgroup

% Keep the short design paragraph together after the two simulation tables.
\begin{samepage}
Table~\ref{tab:simulation-pretrend-tests} evaluates the separate pre-trend tests for the level and response-index leads. We hold $(\kappa,s)=(0.4,0)$ and vary $b_L$ and $b_R$ in the untreated outcome process. Each test assesses whether the two event-time leads at $e=-3,-2$, relative to $e=-1$, are jointly zero at the 5 percent nominal level. In this design, $b_L$ shifts the level leads, whereas the dose-related term $b_R(D-2)$ shifts the response-index leads.
\par
\end{samepage}

Under the null, level rejection rates range from 4.9 to 5.4 percent and response rejection rates from 5.3 to 5.6 percent. With $b_L=0.1$, level-test power rises from 26.4 to 77.8 percent as $n$ increases from 500 to 2,000, and with $b_R=0.1$, response-test power rises from 31.1 to 85.3 percent. Stronger trends increase power further, while the test for the unaffected margin generally remains near its nominal level, with a rejection rate of 6.7 percent for the response test when $b_L=0.1$ and $n=500$.

\section{Application}\label{sec:application}

\subsection{Research design and interpretation}

\citet{bartik-et-al-2019} study the local effects of hydraulic fracturing by comparing counties in the highest prospectivity quartile of a shale play with other counties in the same play. Using their data, \citet{callaway-goodmanbacon-santanna-2024-event} report employment effects for high- and low-dose groups and estimate level effects as a function of dose. We use the CGBS sample to estimate the average employment change relative to counties not yet treated and the association between employment changes and prospectivity within an adoption cohort. We then examine their pre-treatment diagnostics and their contributions to the cellwise continuous-dose regression.

The outcome is log total county employment. The dose is a time-invariant geological prospectivity score, and adoption occurs in the first year of fracking in the county's shale play. The processed CGBS sample is a balanced panel of 402 counties from 1990 through 2014, with 329 positive-dose counties and 73 zero-dose counties. Their preparation code assigns zero to 44 missing prospectivity scores; we retain this coding and examine its consequences below.

Prospectivity scores are constructed separately for each shale play and cannot be compared directly across plays \citep[p.~115]{bartik-et-al-2019}. We retain the CGBS score and read the response index as a statistical association in the recorded score units. Some adoption cohorts contain several plays, so their slopes may reflect differences across plays as well as differences among counties in the same play. Neither cohort centering nor play clustering makes a one-point score difference comparable across plays.

We use the baseline estimators in Section~\ref{sec:baseline-estimation} for seven cohorts, 2001 and 2005--2010, at event times $e=0,\ldots,4$. This gives 307 treated counties and 35 cohort-time cells. Excluding the 2012 cohort holds cohort composition fixed along the event-time path, so the cohort shares and the within-cohort response weights of Corollary~\ref{cor:cohort-interpretations} are common to every event time. Each cell uses the outcome change $Y_{g+e}-Y_{g-1}$. The level contrast uses all counties with $G>g+e$ as controls, between 73 and 368 depending on the cell, whereas the response index uses only cohort-$g$ counties. We aggregate with full-panel cohort shares and include the influence of the estimated shares.

\subsection{Employment paths and separate diagnostics}

Figure~\ref{fig:fracking-margins} reports the event-time estimates. Both are close to zero at adoption and rise over the following four years. At event time four, the level contrast is 0.0474 with an i.i.d. standard error of 0.0116. The response index is 0.0284 log points per prospectivity point with an i.i.d. standard error of 0.0082. Under county-level i.i.d. sampling, joint 95 percent Gaussian max-$t$ intervals over the ten estimates exclude zero for both paths from event time two onward.

\begin{figure}[!htbp]
  \centering
  \caption{Event-time level contrast and response index}
  \label{fig:fracking-margins}
  \vspace{0.35em}
  \includegraphics[width=0.97\textwidth]{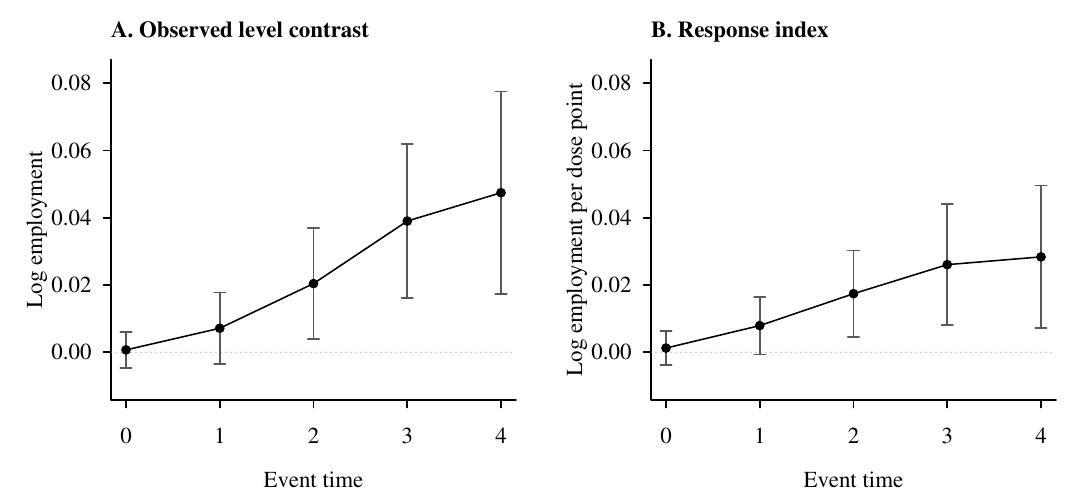}
  \vspace{0.35em}
  \begin{minipage}{0.96\textwidth}
    \begin{singlespace}
      \footnotesize
      \textit{Notes:} The figure reports aggregates with full-panel cohort shares for the 2001 and 2005--2010 cohorts. Vertical bars are joint 95 percent Gaussian max-$t$ intervals over both panels under county-level i.i.d. sampling. Panel~A is the level contrast, which equals the level margin under PT-L; Panel~B is the within-cohort statistical response index in the recorded prospectivity units.
    \end{singlespace}
  \end{minipage}
\end{figure}

For the pre-specified leads $e=-3,-2$, we use $g-1$ as the reference period and $G>g$ controls, as defined in Section~\ref{sec:pretrend}. The two level leads reject a joint zero restriction ($p<0.001$), whereas the response-index leads do not ($p=0.322$). These tests use Corollary~\ref{cor:pretrend-inference} and Theorem~\ref{thm:event-aggregation} under county-level i.i.d. sampling. We therefore report the level path as an observed-data contrast and retain the statistical interpretation of the response path. Nonrejection of its pre-trend restriction establishes neither PT-R nor the absence of selection on gains.

The two diagnostics can diverge because PT-L concerns mean untreated changes across groups, while PT-R concerns their relation to dose within a cohort, and Proposition~\ref{prop:pt-nonnested} shows that neither restriction implies the other. \citet[pp.~125, 130]{bartik-et-al-2019} also report positive employment pre-trends and allow for differential trends and a trend break in their richer specifications. CGBS interpret their dose-bin pre-treatment paths as supporting parallel trends. Those plots, the Bartik specifications, and our baseline contrasts use different samples and comparisons.

\subsection{What the cellwise dose coefficient measures}

Figure~\ref{fig:fracking-mixture} applies Corollary~\ref{cor:share} to the 35 cohort-time cells. The weight $\lambda_{g,t}$ on the response index ranges from 0.019 to 0.117. This weight changes over event time only through the not-yet-treated share. Its mean, weighted by cohort shares, rises from 0.066 at event time zero to 0.084 at event time four, while the weighted number of eligible controls falls from 209.8 to 109.9.

Panel~B aggregates the exact cellwise components using the same cohort shares. At event time four, the aggregate of the cellwise OLS coefficients is 0.0111, with 0.0097 from the level-contrast component and 0.0013 from the response-index component. Thus the coefficient mainly reflects the treated-control contrast per unit of mean dose. Its units of log points per prospectivity point do not make it a within-cohort response. This is the comparison for which the level leads reject.

\begin{figure}[!htbp]
  \centering
  \caption{Cellwise continuous-dose regression weights and components}
  \label{fig:fracking-mixture}
  \vspace{0.35em}
  \includegraphics[width=0.97\textwidth]{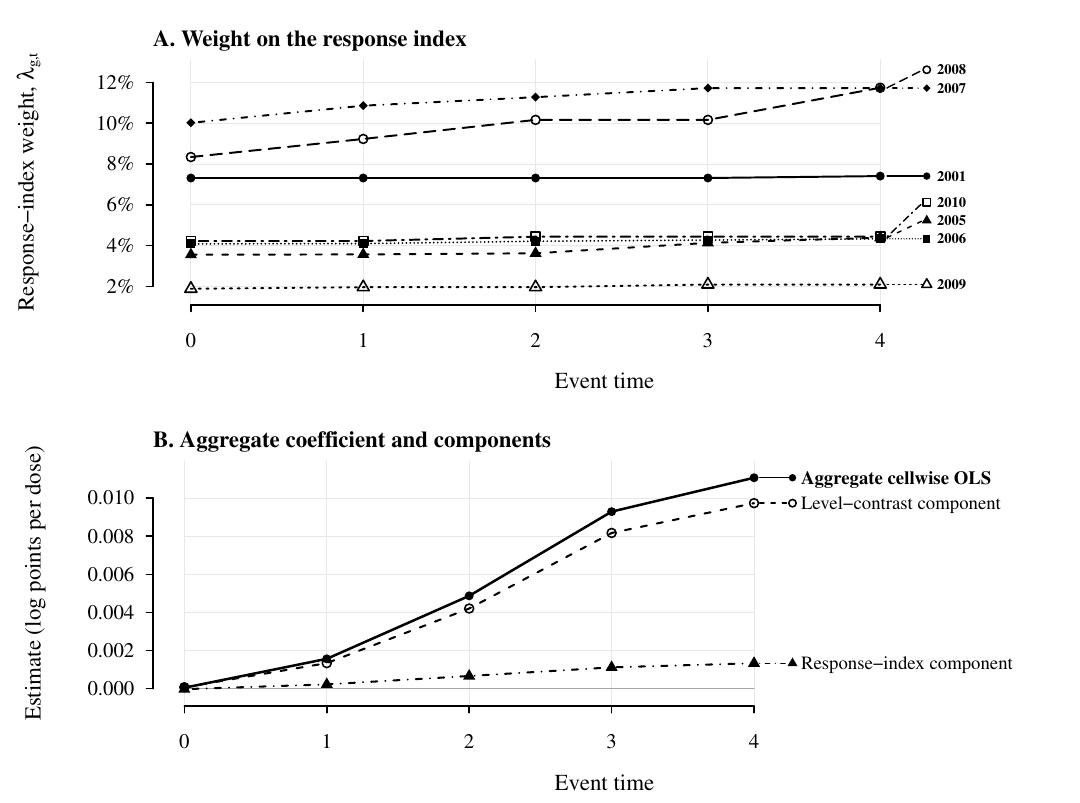}
  \vspace{0.35em}
  \begin{minipage}{0.96\textwidth}
    \begin{singlespace}
      \footnotesize
      \textit{Notes:} Panel~A reports $\lambda_{g,t}$ from \eqref{eq:cell-twfe-mixture} for each cohort-time cell. Panel~B reports cohort-share aggregates of the cellwise coefficient and its two exact components. The identity is applied cell by cell; the figure is not a decomposition of a pooled multi-period staggered TWFE coefficient.
    \end{singlespace}
  \end{minipage}
\end{figure}

We also test whether the five event-time aggregates of $\delta_{L,g,t}-\mu_{D,g}\theta_{R,g,t}$ are jointly zero. The i.i.d. Wald test does not reject ($p=0.228$). This is an aggregate restriction, not a test that the two endpoints coincide in every cell. Appendix~\ref{app:fracking-checks} gives the inference details.

Removing the 44 counties whose missing scores were coded as zero raises the long-window level contrast from 0.0432 to 0.0513. The response estimates are unchanged. These counties enter the control means but not the within-cohort dose covariance, so their removal changes only the level comparison. Appendix~\ref{app:fracking-checks} reports all three windows.

\subsection{Inference and comparison with CGBS}

Table~\ref{tab:fracking-comparison} reports window averages. The short-run and long-run level contrasts are 0.0094 and 0.0432, and the corresponding response indices are 0.0088 and 0.0272 log points per prospectivity point. The larger long-run estimates agree with the employment pattern reported by \citet{bartik-et-al-2019} and \citet{callaway-goodmanbacon-santanna-2024-event}. The CGBS benchmark in Panel~B is a level effect at dose four. Its long-run value of 0.0277 has a different meaning and different units from the response index of 0.0272. Samples, controls, and aggregation also differ, and Appendix~\ref{app:fracking-checks} documents these choices.

\begin{table}[!htbp]
  \centering
  \caption{Level contrasts, response indices, and the CGBS benchmark}
  \label{tab:fracking-comparison}
  {\renewcommand{\arraystretch}{1.15}
    \begin{tabular}{llrrr}
\toprule
Object & Window & Estimate & iid SE & Play SE \\ 
\midrule
\multicolumn{5}{l}{\textit{A. Proposed level contrast and response index}} \\
$\delta_L$ & Short & 0.0094 & 0.0040 & 0.0069 \\
$\delta_L$ & Long & 0.0432 & 0.0101 & 0.0206 \\
$\delta_L$ & Overall & 0.0229 & 0.0063 & 0.0121 \\
$\theta_R$ & Short & 0.0088 & 0.0031 & 0.0063 \\
$\theta_R$ & Long & 0.0272 & 0.0075 & 0.0160 \\
$\theta_R$ & Overall & 0.0162 & 0.0046 & 0.0100 \\
\midrule
\multicolumn{5}{l}{\textit{B. CGBS level at $d=4$}} \\
CGBS $d=4$ & Short & 0.0059 & 0.0061 & 0.0077 \\
CGBS $d=4$ & Long & 0.0277 & 0.0131 & 0.0226 \\
CGBS $d=4$ & Overall & 0.0163 & 0.0089 & 0.0138 \\
\bottomrule
\end{tabular}
}
  \vspace{0.5em}
  \begin{minipage}{0.98\textwidth}
    \begin{singlespace}
      \footnotesize
      \textit{Notes:} Short averages event times 0--2; long averages 3--4; overall averages 0--4. Level contrasts and CGBS estimates are in log points; response indices are in log points per recorded prospectivity point. Panel~A uses 307 treated counties in seven cohorts, with i.i.d. standard errors that include estimation of full-panel cohort shares. In Panel~B, the i.i.d. standard errors include estimation of the zero-dose control mean conditional on the empirical spline basis. Play SE is the CR1 standard error clustered by shale play. Appendix~\ref{app:fracking-checks} gives the CGBS sample details and archived standard errors.
    \end{singlespace}
  \end{minipage}
\end{table}

\looseness=-1 Adoption timing is shared within shale plays, so Table~\ref{tab:fracking-comparison} reports play-clustered CR1 standard errors alongside the standard errors under i.i.d. sampling across counties. The design has 14 shale plays, and several response cohorts contain only one play. Play-clustered intervals for all four short- and long-window estimates in Panel~A include zero. Inference under i.i.d. sampling across counties follows Theorems~\ref{thm:baseline-stacked} and \ref{thm:event-aggregation}, and the clustered results are a sensitivity analysis outside their sampling assumptions. CGBS inference also changes when uncertainty in the estimated control mean and dependence among counties are included, as shown in Appendix~\ref{app:fracking-checks}.

\section{Conclusion}\label{sec:conclusion}

CGBS decompose the continuous-treatment TWFE coefficient, show that its weights on causal responses depend on the dose distribution and the size of the untreated group while selection on gains enters alongside, and recommend estimating treatment effects and causal responses at each dose instead. We take a different route through the same coefficient. In each cohort-time cell it is an exact convex combination of two comparisons that applied researchers already run: the difference in mean outcome changes between the cohort and its not-yet-treated controls, and the OLS slope of the outcome change on dose within the cohort. Our suggestion is to stop reporting their mixture and report the two comparisons separately.

Separating them pays in two ways. The first is identification. The level margin needs a control group and PT-L. The response index uses only the treated cohort and needs, for a causal reading, PT-R and a restriction on selection on gains. Because neither restriction implies the other, evidence against one leaves the other comparison intact, and each has its own pre-trend diagnostic. The second is estimation. A difference in means and an OLS slope involve no dose density, derivative estimator, bandwidth, or functional form for the dose-response curve, and their influence functions are explicit, so joint inference across cells, event times, and margins is closed form in the baseline case and carries over to the covariate-adjusted estimators. The price is that the response index is one number with fixed weights rather than a dose-response curve, and its causal content is only as good as the response restrictions behind it.

The hydraulic fracturing application illustrates the empirical relevance of the distinction. The level leads reject a joint zero restriction, whereas the response-index leads do not. The latter remain a diagnostic rather than evidence for a causal response interpretation. The exact cellwise mixture places most of its weight on the level-contrast component, and removing counties with missing prospectivity changes the level estimates but leaves the response estimates unchanged. These findings show what is lost when the two sources of variation are summarized by a single continuous-dose coefficient.

\clearpage
\addcontentsline{toc}{section}{References}
\begin{singlespace}
  \bibliographystyle{ecta}
  \bibliography{references}
\end{singlespace}

\clearpage
\appendix

% Number equations and each theorem-like environment independently within
% appendix sections: A.1, A.2, ..., B.1, B.2, ....
\counterwithin{equation}{section}
\renewcommand{\theHequation}{appendix.\Alph{section}.\arabic{equation}}
\counterwithin{theorem}{section}
\renewcommand{\theHtheorem}{appendix.\Alph{section}.\arabic{theorem}}
\counterwithin{proposition}{section}
\renewcommand{\theHproposition}{appendix.\Alph{section}.\arabic{proposition}}
\counterwithin{lemma}{section}
\renewcommand{\theHlemma}{appendix.\Alph{section}.\arabic{lemma}}
\counterwithin{corollary}{section}
\renewcommand{\theHcorollary}{appendix.\Alph{section}.\arabic{corollary}}
\counterwithin{assumption}{section}
\renewcommand{\theHassumption}{appendix.\Alph{section}.\arabic{assumption}}
\counterwithin{definition}{section}
\renewcommand{\theHdefinition}{appendix.\Alph{section}.\arabic{definition}}
\counterwithin{example}{section}
\renewcommand{\theHexample}{appendix.\Alph{section}.\arabic{example}}
\counterwithin{remark}{section}
\renewcommand{\theHremark}{appendix.\Alph{section}.\arabic{remark}}

\section{Additional results and proofs}\label{app:mainproofs}

\subsection{Identification}

\begin{proof}[Proof of Proposition~\ref{prop:pt-nonnested}]
  For PT-R without PT-L, let $U=1$ for every treated unit and $U=0$ for every control, with any nondegenerate positive-dose distribution. Then $\E[U\mid D,T=1]=1=\E[U\mid T=1]$, but the treated and control means of $U$ differ.

  For PT-L without PT-R, let controls have $U=0$.  Among treated units, let $P(D=1\mid T=1)=P(D=2\mid T=1)=1/2$, with $U=1$ at $D=1$ and $U=-1$ at $D=2$.  The treated and control means of $U$ are both zero, but the conditional treated mean varies with dose.  Finally, under PT-R, iterated expectations give $\E[U\mid D,T=1]=\E[U\mid T=1]$; PT-L equates the latter with $\E[U\mid T=0]$, proving \eqref{eq:all-dose-pt}.
\end{proof}

\begin{proof}[Proof of Theorem~\ref{thm:two-margins}]
  For part (i), consistency gives $Y=U+\tau(D)$ for treated units and $Y=U$ for untreated units.  PT-L therefore gives
  \[
    \E[Y\mid T=1]-\E[Y\mid T=0]
    =
    \E_+[\tau(D)]=\tau_L.
  \]
  For part (ii), the two identities in \eqref{eq:baseline-balancing} follow immediately from $\alpha(D)=(D-\mu_D)/V_D$.  The OLS normal equations yield $\theta_R=\Cov_+(D,Y)/V_D$ and show invariance to an additive constant.  For part (iii), consistency gives
  \[
    \Cov_+(D,Y)
    =
    \Cov_+(D,U)+\Cov_+\{D,\tau(D)\}.
  \]
  PT-R makes the first covariance zero, proving \eqref{eq:response-gain-id}.
\end{proof}

\begin{proof}[Proof of Theorem~\ref{thm:staggered}]
  Assumption~\ref{ass:staggered} gives $Y_c=U_c+\tau_c(D)$ for $G=g$ and $Y_c=U_c$ for $G>t$. $\mathrm{PT\text{-}L}_{g,t}$ therefore gives
  \[
    \E[Y_c\mid G=g]-\E[Y_c\mid G>t]
    =
    \E[\tau_c(D)\mid G=g]=\tau_{L,g,t}.
  \]
  Likewise,
  \[
    \Cov(D,Y_c\mid G=g)
    =
    \Cov(D,U_c\mid G=g)
    +\Cov\{D,\tau_c(D)\mid G=g\}.
  \]
  $\mathrm{PT\text{-}R}_{g,t}$ sets the first covariance to zero.  Division by $V_g$ proves \eqref{eq:cohort-response-gain-id}.
\end{proof}

\subsection{The regression mixture}

\begin{proof}[Proof of Theorem~\ref{thm:mixture}]
  Because $D=0$ when $T=0$,
  \[
    \E[D]=p\mu_D,
    \qquad
    \E[D^2]=p(V_D+\mu_D^2),
  \]
  and hence
  \begin{equation}\label{eq:twfe-var-proof}
    \Var(D)=p\{V_D+\rho\mu_D^2\}.
  \end{equation}
  Writing $\mu_{Y,1}=\E[Y\mid T=1]$ and $\mu_{Y,0}=\E[Y\mid T=0]$,
  \[
    \E[DY]=p\{\Cov_+(D,Y)+\mu_D\mu_{Y,1}\}.
  \]
  Subtracting $\E[D]\E[Y]$ yields
  \begin{equation}\label{eq:twfe-cov-proof}
    \Cov(D,Y)
    =
    p\{V_D\theta_R+\rho\mu_D\Delta_Y\}.
  \end{equation}
  The ratio of \eqref{eq:twfe-cov-proof} and \eqref{eq:twfe-var-proof} is \eqref{eq:share-formula}; collecting the two numerator terms produces $\lambda$.  PT-L gives $\Delta_Y=\tau_L$ by Theorem~\ref{thm:two-margins}(i).  Direct differentiation of the ratio in \eqref{eq:share-formula} gives \eqref{eq:share-derivative}.

  First differencing a two-period regression with unit and time effects turns the regressor $D_i\ind\{s=2\}$ into $D_i$ and the outcome into $Y_i$; the resulting regression includes a constant.  Thus its coefficient is the slope used above.
\end{proof}

\begin{proof}[Proof of Corollary~\ref{cor:share}]
  Repeat the population-moment calculation in the proof of Theorem~\ref{thm:mixture} under $P_c$.  Under $P_c$, $G=g$ is the treated group, $G>t$ is the control group, $D_c$ equals $D$ for cohort-$g$ units and zero for controls, and $\Delta_Y=\delta_{L,g,t}$.  The corollary assumes the positive group probabilities, finite second moments, and positive within-cohort dose variation required by that calculation. Substituting $\rho_{g,t}$ from \eqref{eq:cell-twfe-components} gives \eqref{eq:cell-twfe-mixture}. Under Assumption~\ref{ass:staggered} and the conditions of Theorem~\ref{thm:staggered}(i), $\delta_{L,g,t}=\tau_{L,g,t}$.
\end{proof}

\begin{proof}[Proof of the identities in Remark~\ref{rem:cgbs-weights}]
  Because $D=0$ when $T=0$, $\E[D]=p\mu_D$ and $\Var(D)=p(V_D+\rho\mu_D^2)$ by \eqref{eq:twfe-var-proof}.  For $d>0$, only treated units satisfy $D\ge d$, so
  \[
    \E[(D-\E[D])\ind\{D\ge d\}]
    =
    p\,\E_+[(D-\mu_D)\ind\{D\ge d\}]
    +p\rho\mu_D\,\Q(D\ge d).
  \]
  Since $\E_+[D-\mu_D]=0$, the first expectation equals $-\E_+[(D-\mu_D)\ind\{D<d\}]$; for an absolutely continuous treated dose this is $V_DW(d)$ and $\Q(D\ge d)=\Q(D>d)$.  Dividing by $\Var(D)$ and using $1-\lambda=\rho\mu_D^2/(V_D+\rho\mu_D^2)$ gives the first identity in \eqref{eq:cgbs-weight-identity}.  For the second, $(\mu_D-\E[D])pa=p\rho\mu_Da$, and division by $\Var(D)$ gives $\rho\mu_Da/(V_D+\rho\mu_D^2)=(1-\lambda)a/\mu_D$.

  Under the hypotheses of Theorem~\ref{thm:response-decomp}, $q$ is continuously differentiable on $[a,b]$, so Fubini's theorem gives
  \[
    \int_a^b\Q(D>d)\,q'(d)\dd d
    =
    \E_+\!\left[\int_a^D q'(d)\dd d\right]
    =
    \E_+[q(D)]-q(a)
    =
    \tau_L-q(a),
  \]
  where the last equality uses \eqref{eq:diagonal-compatibility} and iterated expectations.  By \eqref{eq:qprime-representation}, the first component of $w_1$ contributes $\lambda\theta_R$, and the lowest-dose term contributes $w_0q(a)/a=(1-\lambda)q(a)/\mu_D$.  The three contributions sum to $\lambda\theta_R+(1-\lambda)\tau_L/\mu_D$, which is \eqref{eq:share-formula} with $\Delta_Y=\tau_L$ under Assumption~\ref{ass:levelpt}.
\end{proof}

\subsection{Response weights and selection}

\begin{lemma}[Signed-measure integration identity]\label{lem:signed-connected}
  Let $\nu$ be a finite signed Borel measure supported on $[a,b]$ with $\nu([a,b])=0$, and define $W(t)=-\nu(( -\infty,t])$.  If $q$ is absolutely continuous and $\int_a^b|W(t)q'(t)|\dd t<\infty$, then
  \[
    \int q(d)\nu(\dd d)=\int_a^bW(t)q'(t)\dd t.
  \]
\end{lemma}

\begin{proof}
  Write $q(d)=q(a)+\int_a^b\ind\{t<d\}q'(t)\dd t$.  The constant integrates to zero.  Fubini's theorem gives
  \[
    \int q(d)\nu(\dd d)
    =
    \int_a^bq'(t)\nu((t,b])\dd t.
  \]
  Zero total mass implies $\nu((t,b])=-\nu([a,t])=W(t)$ apart from an immaterial atom convention on a Lebesgue-null set, proving the identity.
\end{proof}

\begin{proof}[Proof of Lemma~\ref{lem:weight-nonnegative}]
  If $d<\mu_D$, then $(D-\mu_D)\ind\{D\le d\}\le0$.  If $d\ge\mu_D$, centeredness gives
  \[
    \E_+[(D-\mu_D)\ind\{D\le d\}]
    =
    -\E_+[(D-\mu_D)\ind\{D>d\}]\le0.
  \]
  Multiplication by $-1/V_D$ proves nonnegativity.
\end{proof}

\begin{proof}[Proof of Theorem~\ref{thm:response-decomp}]
  Under Assumptions~\ref{ass:consistency}, \ref{ass:dosept}, and~\ref{ass:connected}(iii),
  \[
    \begin{aligned}
      \E[Y\mid D,T=1]
      &=\E[U\mid T=1]+\E[\tau(D)\mid D,T=1]\\
      &=\E[U\mid T=1]+q(D)
      \quad \Q\text{-almost surely}.
    \end{aligned}
  \]
  By iterated expectations and $\E_+[\alpha]=0$,
  \[
    \theta_R
    =
    \E_+[\alpha(D)\tau(D)]
    =
    \E_+[\alpha(D)q(D)].
  \]
  Under Assumption~\ref{ass:connected}, $q$ is continuously differentiable on $[a,b]$, hence absolutely continuous, and $|W(d)|\le\{\E_+[\alpha(D)^2]\}^{1/2}=V_D^{-1/2}$ by Cauchy--Schwarz, so $\int_a^b|W(t)q'(t)|\dd t<\infty$.  Apply Lemma~\ref{lem:signed-connected} to the signed measure $\nu(B)=\E_+[\alpha(D)\ind\{D\in B\}]$.  Its cumulative weight is $W$, giving \eqref{eq:qprime-representation}.  Apply the same lemma with $q(d)=d$.  The left side is $\E_+[\alpha D]=1$ and the derivative is one, proving \eqref{eq:W-normalized}.  Nonnegativity follows from Lemma~\ref{lem:weight-nonnegative}.

  The chain rule along the diagonal gives
  \[
    q'(d)=ACRT(d\mid d)+S(d).
  \]
  Substitution proves \eqref{eq:selection-decomp} and \eqref{eq:selection-part}; the integrals are finite because $W$ is bounded and $ACRT(\cdot\mid\cdot)$ and $S$ are continuous on the compact interval $[a,b]$.  The sign results follow from $W\ge0$, and the triangle inequality gives \eqref{eq:sensitivity-bound}.
\end{proof}

\subsection{Balancing and local reweighting}\label{app:response-characterizations}

\begin{proof}[Proof of Proposition~\ref{prop:minnorm}]
  For any $a\in\mathcal A_R$,
  \[
    1=\E_+[a(D)D]=\E_+[a(D)(D-\mu_D)].
  \]
  Cauchy--Schwarz gives $1\le\{\E_+[a(D)^2]V_D\}^{1/2}$, hence $\E_+[a(D)^2]\ge1/V_D$.  The function $\alpha=(D-\mu_D)/V_D$ belongs to the class and attains equality. Equality in Cauchy--Schwarz requires proportionality to $D-\mu_D$, and the normalization fixes the constant at $1/V_D$.  Under homoskedasticity,
  \[
    \Var_+\{a(D)\varepsilon\}
    =
    \sigma^2\E_+[a(D)^2],
  \]
  so the same unique minimizer applies.
\end{proof}

Proposition~\ref{prop:minnorm} gives $\alpha$ the smallest $L_2(\Q_D)$ norm among representers satisfying the two restrictions in \eqref{eq:balancing-class-main}.  The same balancing geometry appears in the least-squares estimands of \citet{hines-diazordaz-vansteelandt-2026}.  Without covariates, their two least-squares targets coincide with $\theta_R$, and under outcome homoskedasticity their efficiency-optimal representer reduces to $\alpha$.  Their criterion minimizes a nonparametric efficiency bound for a sample-analogue target, whereas Proposition~\ref{prop:minnorm} minimizes the $L_2(\Q_D)$ norm directly.

\begin{proof}[Proof of Proposition~\ref{prop:tilt}]
  For $|\varepsilon|\le\varepsilon_0/4$, the integrands $\mu(D)e^{\varepsilon D}$, $De^{\varepsilon D}$, and $e^{\varepsilon D}$, together with their $\varepsilon$-derivatives, are dominated by $\{1+|\mu(D)|\}(1+D^2)e^{\varepsilon_0|D|/4}$.  This bound is $\Q$-integrable: by Cauchy--Schwarz its mean is at most $\{\E_+[(1+|\mu(D)|)^2]\}^{1/2}\{\E_+[(1+D^2)^2e^{\varepsilon_0|D|/2}]\}^{1/2}$, where $\E_+[\mu(D)^2]\le\E_+[Y^2]<\infty$ by Jensen's inequality and $(1+D^2)^2e^{\varepsilon_0|D|/2}\le Ce^{\varepsilon_0|D|}$ for a constant $C$.  Dominated convergence therefore justifies differentiating $M(\varepsilon)$, $\mathcal Y(\varepsilon)$, and $\mathcal D(\varepsilon)$ under the integral near zero.  At $\varepsilon=0$, the derivative of the likelihood ratio in \eqref{eq:tilt-law} is $d-\mu_D$, and differentiation under the integral gives
  \[
    \mathcal Y'(0)
    =
    \E_+[(D-\mu_D)\mu(D)]
    =
    \Cov_+(D,Y),
  \]
  and
  \[
    \mathcal D'(0)
    =
    \E_+[(D-\mu_D)D]=V_D.
  \]
  Taking the ratio proves \eqref{eq:tilt-interpretation}.
\end{proof}

Exponential tilting also appears in the stochastic policy framework of \citet{jetsupphasuk-et-al-2025}, but the estimands are different.  For a finite tilt parameter, they define a counterfactual dose distribution and, under their maintained causal identification assumptions, identify an average potential-outcome contrast under that policy relative to no treatment.  Here the tilted family $\{\Q_\varepsilon\}$ defines only a differentiable path through the observed treated-dose distribution, and $\theta_R$ is the ratio of the derivatives of the reweighted outcome and dose means at $\varepsilon=0$.  We therefore neither interpret $\mathcal Y(\varepsilon)$ at a finite $\varepsilon$ as a counterfactual policy value nor identify a finite stochastic policy effect.

\begin{remark}[Comparison with a global average derivative]\label{rem:global-average-derivative}
  The global average derivative studied by \citet{callaway-goodmanbacon-santanna-2026} is a different response parameter.  On a fixed compact support, integration by parts in its density-score representation can introduce endpoint terms; \citet{newey-stoker-1993} discuss the associated difficulty of boundary evaluation in an unrestricted fixed-support model.  Weighted average derivatives do not in general require density estimation: the least-squares weighted average derivatives of \citet{hines-diazordaz-vansteelandt-2026} form a density-free class that includes the representer used here, which is why $\theta_R$ can be estimated without a dose density, a derivative estimator, or a bandwidth.  We make no claim about the pathwise regularity of the global derivative under alternative support models.
\end{remark}

\subsection{Cohort-specific representations}

\begin{proof}[Proof of Corollary~\ref{cor:cohort-interpretations}]
  Apply Propositions~\ref{prop:minnorm} and \ref{prop:tilt} and Lemma~\ref{lem:weight-nonnegative} under $P(\,\cdot\mid G=g)$, with outcome $Y_c$ and the conditions stated for each result.  For the selection representation, repeat the proof of Theorem~\ref{thm:response-decomp} under $P(\,\cdot\mid G=g)$.  Theorem~\ref{thm:staggered}(ii) supplies the gain-covariance representation and the required cohort moments, while Assumption~\ref{ass:connected}(iii) for this cohort gives $q_{g,t}(D)=\E[\tau_c(D)\mid D,G=g]$ almost surely.  Iterated expectations, the signed-measure identity, and the chain rule then give the stated representation using only the cohort distribution.  Because the distribution of $D\mid G=g$ does not vary with $t$, neither do the representers or weights.
\end{proof}

\section{Proofs for closed-form stacked inference and aggregation}\label{app:baseline-inference}

\begin{proof}[Proof of Theorem~\ref{thm:baseline-stacked}]
  For the level estimator, the standard ratio expansion for a conditional mean gives
  \[
    \sqrt n(\overline Y_{c,g}-\mu_{1,g,t})
    =
    \frac1{\sqrt n}\sum_{i=1}^n
    \frac{\ind\{G_i=g\}}{q_g}(Y_{c,i}-\mu_{1,g,t})
    +o_{\Prb}(1),
  \]
  with the analogous expansion for the $G>t$ mean.  Subtraction yields \eqref{eq:phiL-0}.

  For the response slope, work first under $P(\,\cdot\mid G=g)$ and write $\mu_D=\mu_{D,g}$, $\mu_Y=\mu_{1,g,t}$, $C=\E[(D-\mu_D)(Y_c-\mu_Y)\mid G=g]$, and $V=V_g$. The influence functions of $C$ and $V$ under this conditional distribution are
  \[
    (D-\mu_D)(Y_c-\mu_Y)-C
    \quad\text{and}\quad
    (D-\mu_D)^2-V.
  \]
  The ratio rule for $C/V=\theta_{R,g,t}$ gives
  \[
    \frac1V\left[
      (D-\mu_D)(Y_c-\mu_Y)
      -\theta_{R,g,t}(D-\mu_D)^2
      \right].
  \]
  Transporting this cohort-specific influence function to the full-panel distribution $P$ multiplies it by $\ind\{G=g\}/q_g$, producing \eqref{eq:phiR-0}.

  The finite second moments in Assumption~\ref{ass:baseline-inference} justify these smooth sample-moment expansions.  Stack the finitely many cells.  Each fixed linear combination of the stacked influence function vector is i.i.d., mean zero, and square integrable, so the Cram\'er--Wold device and the scalar central limit theorem give \eqref{eq:baseline-vector-clt}.

  The sample proportions, conditional means, dose means, variances, and slopes are consistent.  Under the displayed moment conditions, replacing them in the influence functions yields $L_2$-consistent columns, and hence
  \[
    \mathbb P_n[
    \widehat{\boldsymbol\phi}^{P}
    \widehat{\boldsymbol\phi}^{P\prime}]
    -
    \mathbb P_n[
    \boldsymbol\phi^{P}
    \boldsymbol\phi^{P\prime}]
    =o_{\Prb}(1)
  \]
  entry by entry.  The law of large numbers proves \eqref{eq:baseline-sigma}.
\end{proof}

\begin{proof}[Proof of Theorem~\ref{thm:event-aggregation}]
  For each aggregate $r$, add and subtract the population-weighted cell estimates to obtain
  \[
    \widehat\eta_r-\eta_r
    =
    \sum_{j\in\mathcal J_r}
    w_{rj}(\widehat\vartheta_j-\vartheta_j)
    +
    \sum_{j\in\mathcal J_r}
    \vartheta_j(\widehat w_{rj}-w_{rj})
    +
    \sum_{j\in\mathcal J_r}
    (\widehat w_{rj}-w_{rj})
    (\widehat\vartheta_j-\vartheta_j).
  \]
  The first two sums are jointly $O_{\Prb}(n^{-1/2})$ by the maintained asymptotic linear representations.  Because the family of aggregates and all coordinate sets are fixed and finite, every term in the last sum is $O_{\Prb}(n^{-1})$, and hence that sum is $o_{\Prb}(n^{-1/2})$.  Substituting the two first-order expansions gives the influence function in \eqref{eq:event-if}.  The joint influence vector is i.i.d., mean zero, and square integrable, so the Cram\'er--Wold device yields the stated multivariate central limit theorem.

  For covariance consistency, the assumed empirical $L_2$ convergence of the combined plug-in columns and Cauchy--Schwarz imply, entry by entry,
  \[
    \mathbb P_n[
    \widehat{\boldsymbol\psi}
    \widehat{\boldsymbol\psi}']
    -
    \mathbb P_n[
    \boldsymbol\psi
    \boldsymbol\psi']
    =o_{\Prb}(1).
  \]
  The law of large numbers then gives \eqref{eq:event-sigma}.

  For the cohort-share construction, write
  $\widehat q_g=\mathbb P_n\ind\{G=g\}$ and
  $\widehat Q_e=\mathbb P_n\ind\{G\in\mathcal G_e\}$.  The ratio delta method applied to $\widehat w_{g,e}=\widehat q_g/\widehat Q_e$ gives
  \[
    \sqrt n(\widehat w_{g,e}-w_{g,e})
    =
    \frac1{\sqrt n}\sum_{i=1}^n
    \frac{
    \ind\{G_i=g\}
    -w_{g,e}\ind\{G_i\in\mathcal G_e\}
    }{Q_e}
    +o_{\Prb}(1),
  \]
  which is \eqref{eq:event-weight-alr}.  The positive lower bound on $Q_e$ and finiteness of the fixed cohort family also give square integrability and consistency of the displayed plug-in weight influence functions.
\end{proof}

\begin{proof}[Proof of Corollary~\ref{cor:pretrend-inference}]
  Apply the sample-mean and slope expansions in the proof of Theorem~\ref{thm:baseline-stacked} to $Y_c^{\mathrm{pre}}$ and the disjoint groups $G=g$ and $G>g$. These expansions use only the observed moments and positive denominators, so the ordering of the two outcome dates is immaterial. The stated moments give square-integrable influence functions and empirical $L_2$ consistency of their plug-in columns. Stacking the fixed finite collection, alone or together with the post-treatment collection under Assumption~\ref{ass:baseline-inference}, and applying the same central limit theorem and covariance argument proves the claim. Theorem~\ref{thm:event-aggregation} then supplies the aggregate influence functions and covariance under its conditions.
\end{proof}

\section{Covariate-adjusted results, DML derivations, and proofs}\label{app:dmlproofs}

This appendix develops the population characterizations and semiparametric details summarized in Section~\ref{sec:covariates}; adjusted response objects carry the superscript $(X)$.

\subsection{Adjusted balancing, tilt, weights, and selection}

Conditional on $G=g$, $\theta_{R,g,t}^{(X)}$ is the least-squares estimand $\Psi$ of \citet{hines-diazordaz-vansteelandt-2026}: the ratio of the average conditional covariance of dose and outcome to the average conditional dose variance, or equivalently an average of conditional least-squares slopes weighted by conditional dose variance.  Define its representer by
\begin{equation}\label{eq:alpha-X}
  \alpha_g^{(X)}(d,x)=\frac{d-e_g(x)}{V_g^{(X)}}.
\end{equation}
Then $\theta_{R,g,t}^{(X)}=\E[\alpha_g^{(X)}(D,X)Y_c\mid G=g]$.

Define
\[
  \mathcal A_g^{(X)}
  =
  \left\{
  \alpha\in L_2(P_{D,X\mid G=g}):
  \E[\alpha(D,X)\mid G=g,X]=0,
  \quad
  \E[\alpha(D,X)D\mid G=g]=1
  \right\}.
\]

\begin{proposition}[Adjusted minimum-norm balancing weight]\label{prop:covariate-properties}
  If $V_g^{(X)}>0$, then $\alpha_g^{(X)}$ is the unique minimizer of
  \[
    \min_{\alpha\in\mathcal A_g^{(X)}}
    \E[\alpha(D,X)^2\mid G=g],
    \qquad
    \min_{\alpha\in\mathcal A_g^{(X)}}
    \E[\alpha(D,X)^2\mid G=g]
    =
    \frac1{V_g^{(X)}}.
  \]
\end{proposition}

The local-tilt interpretation has an analogue that holds the covariate distribution fixed.  Exponentially tilt each conditional distribution of $D\mid G=g,X=x$ by $\exp(\varepsilon d)$ while holding the distribution of $X\mid G=g$ fixed, and let $\mathcal Y_{g,t}^{(X)}(\varepsilon)$ and $\mathcal D_g^{(X)}(\varepsilon)$ denote the induced cohort-average outcome and dose.  Suppose $V_g^{(X)}>0$, $\E[Y_c^2\mid G=g]<\infty$, and, for some $\varepsilon_0>0$, there is a function $B_g\in L_2(P_{X\mid G=g})$ such that, almost surely,
\[
  \sup_{|\varepsilon|\le\varepsilon_0}
  \frac{\E[(1+|Y_c|)(1+D^2)\exp(\varepsilon D)\mid G=g,X]}
       {\E[\exp(\varepsilon D)\mid G=g,X]}
  \le B_g(X)<\infty.
\]
This common exponential envelope permits differentiation conditional on $X$ and supplies integrable domination when averaging over $X\mid G=g$.  The dominated convergence argument in the proof of Proposition~\ref{prop:tilt} then gives
\[
  \theta_{R,g,t}^{(X)}
  =
  \frac{\mathcal Y_{g,t}^{(X)\prime}(0)}
  {\mathcal D_g^{(X)\prime}(0)}.
\]
Indeed, the conditional likelihood-ratio derivative at zero is $d-e_g(x)$, so differentiation yields the numerator and denominator in \eqref{eq:cohort-response-X-def}.

Define the adjusted response weight
\[
  W_g^{(X)}(d,x)
  =
  -\E[\alpha_g^{(X)}(D,X)\ind\{D\le d\}\mid G=g,X=x].
\]
For the derivative interpretation, define
\[
  \tau_{g,t,x}(u\mid d)
  =
  \E[Y_t(g,u)-Y_t(\infty,0)\mid D=d,G=g,X=x],
  \qquad
  q_{g,t,x}(d)=\tau_{g,t,x}(d\mid d).
\]
Assume that, for almost every $x$ under $P(\,\cdot\mid G=g)$, the conditional support is a compact interval $[a_g(x),b_g(x)]$ and the effect surface has a version that is jointly measurable in $(u,d,x)$ and continuously differentiable in $(u,d)$ near its diagonal.  Require the chosen version to satisfy
\begin{equation}\label{eq:diagonal-compatibility-X}
  q_{g,t,X}(D)=\E[\tau_c(D)\mid D,X,G=g]
  \qquad P(\,\cdot\mid G=g)\text{-almost surely}.
\end{equation}
The diagonal is then absolutely continuous.  Assume also that the weighted derivatives below are absolutely integrable after averaging over $X\mid G=g$.  The weights satisfy
\[
  W_g^{(X)}(d,x)\ge0,
  \qquad
  \int_{a_g(x)}^{b_g(x)}W_g^{(X)}(d,x)\dd d
  =
  \frac{\Var(D\mid G=g,X=x)}{V_g^{(X)}},
\]
and therefore
\begin{equation}\label{eq:cohort-W-X-global}
  \E_{X\mid G=g}
  \int W_g^{(X)}(d,X)\dd d=1.
\end{equation}
Normalization holds only after averaging over the cohort covariate distribution; the conditional integral need not equal one.

Define $ACRT_{g,t,x}(d\mid d)$ and $S_{g,t,x}(d)$ as in \eqref{eq:acrt-selection}.  Under consistency, no anticipation, Assumption~\ref{ass:covariate-identification}(iii), and the preceding support, smoothness, compatibility, and integrability conditions,
\begin{equation}\label{eq:cohort-selection-decomp-X}
  \theta_{R,g,t}^{(X)}
  =
  \E_{X\mid G=g}\int W_g^{(X)}(d,X)
  \{ACRT_{g,t,X}(d\mid d)+S_{g,t,X}(d)\}\dd d.
\end{equation}
The no-selection condition is untestable from the two-period comparison, so sign restrictions and the conditional analogue of \eqref{eq:sensitivity-bound} provide a sensitivity analysis.  Because the same $X$ is used across $t$, $e_g$, $V_g^{(X)}$, $\alpha_g^{(X)}$, and $W_g^{(X)}$ are common across outcome dates within a cohort.  As in Section~\ref{sec:why-two}, the derivative interpretation is limited to the connected-support conditions stated above.

\subsection{A compact partialling-out identity}

The distinction between level and response variation survives partialling out $X$.  For a generic two-period model with covariates, write
\begin{align*}
  \pi(x)    & =P(T=1\mid X=x),                      \\
  \mu_D(x)  & =\E[D\mid T=1,X=x],                   \\
  V_D(x)    & =\Var(D\mid T=1,X=x),                 \\
  C_{DY}(x) & =\Cov(D,Y\mid T=1,X=x),               \\
  \Delta_Y(x)
            & =\E[Y\mid T=1,X=x]-\E[Y\mid T=0,X=x].
\end{align*}
Let $\beta_F$ be the coefficient after partialling out the closed linear span of all square-integrable functions of $X$, and define
\begin{align*}
  \beta_F
    & =
  \frac{\E[\Cov(D,Y\mid X)]}{\E[\Var(D\mid X)]}, \\
  A & =\E[\pi(X)V_D(X)],
  \qquad
  B=\E[\pi(X)\{1-\pi(X)\}\mu_D(X)^2].
\end{align*}
When $A>0$, let $\theta_R^{(X)}=\E[\pi(X)C_{DY}(X)]/A$, equivalently the residual projection defined in \eqref{eq:thetaR-X-def} below.

\begin{proposition}[Adjusted partialling-out identity]\label{prop:adjusted-mixture}
  If $D=0$ for $T=0$ and $A>0$, then
  \begin{equation}\label{eq:beta-decomp-raw}
    \beta_F
    =
    \frac{
      A\theta_R^{(X)}
      +\E[\pi(X)\{1-\pi(X)\}\mu_D(X)\Delta_Y(X)]
    }{A+B}.
  \end{equation}
  Under consistency and conditional PT-L, $\Delta_Y(X)=\E[\tau(D)\mid T=1,X]$.
\end{proposition}

Partialling out $X$ therefore does not isolate the adjusted response index: the regression retains a level component.

\subsection{Canonical gradients and generic-cell DML}

Under the treated-group distribution $\Q$, write
\begin{equation}\label{eq:thetaR-X-def}
  \begin{aligned}
    e(x)              & =\E[D\mid T=1,X=x],                        \\
    g(x)              & =\E[Y\mid T=1,X=x],                        \\
    V^{(X)}           & =\E_+[(D-e(X))^2],                         \\
    \theta_R^{(X)}
                      & =\frac{\E_+[(D-e(X))\{Y-g(X)\}]}{V^{(X)}}, \\
    \alpha^{(X)}(d,x) & =\frac{d-e(x)}{V^{(X)}}.
  \end{aligned}
\end{equation}

For generic nuisances in cell notation $(Y,T,D,X)$, define the orthogonal scores
\begin{equation}\label{eq:level-score}
  \psi_L(O;\delta,\bar m,\bar\pi)
  =
  T\{Y-\bar m(X)-\delta\}
  -(1-T)\frac{\bar\pi(X)}{1-\bar\pi(X)}
  \{Y-\bar m(X)\},
\end{equation}
and
\begin{equation}\label{eq:response-score}
  \psi_R^{(X)}(O;\theta,\bar e,\bar g)
  =
  T\left[
    \{D-\bar e(X)\}\{Y-\bar g(X)\}
    -\theta\{D-\bar e(X)\}^2
    \right].
\end{equation}

In generic cell notation, let $p=P_c(T=1)$ and $V^{(X)}=\E[(D-e(X))^2\mid T=1]$, and define the adjusted level functional of the observed distribution
\[
  \delta_L^{(X)}=\E[Y-m_0(X)\mid T=1].
\]
The canonical gradients under $P_c$ are
\begin{equation}\label{eq:phiL}
  \phi_L^{P_c}(O)
  =
  \frac1p\left[
    T\{Y-m_0(X)-\delta_L^{(X)}\}
    -(1-T)\frac{\pi(X)}{1-\pi(X)}\{Y-m_0(X)\}
    \right],
\end{equation}
and
\begin{equation}\label{eq:phiR}
  \phi_R^{(X),P_c}(O)
  =
  \frac{T}{pV^{(X)}}
  \left[
    \{D-e(X)\}\{Y-g(X)\}
    -\theta_R^{(X)}\{D-e(X)\}^2
    \right].
\end{equation}

\begin{proposition}[Canonical gradients]\label{prop:gradients}
  Under overlap, positive residual dose variation, and square integrability of the displayed functions, \eqref{eq:phiL} and \eqref{eq:phiR} are the canonical gradients of the adjusted level and response functionals of the observed distribution in the unrestricted nonparametric model.
\end{proposition}

\begin{theorem}[Generic adjusted joint DML]\label{thm:joint-dml}
  Let $n_c=\sum_iS_{c,i}$, let $\boldsymbol\phi_c^{P_c}=(\phi_L^{P_c},\phi_R^{(X),P_c})'$ with $\Sigma_c=\E[\boldsymbol\phi_c^{P_c}\boldsymbol\phi_c^{P_c\prime}\mid S_c=1]$, and let $\widehat{\boldsymbol\phi}_c^{P_c}$ denote the plug-in gradient obtained from \eqref{eq:phiL} and \eqref{eq:phiR} by replacing $(m_0,\pi,e,g)$ with their out-of-fold estimates and $(p,V^{(X)},\delta_L^{(X)},\theta_R^{(X)})$ with $(\mathbb P_{n,c}[T],\mathbb P_{n,c}[T\{D-\widehat e_{-k}(X)\}^2]/\mathbb P_{n,c}[T],\widehat\delta_L^{(X)},\widehat\theta_R^{(X)})$.  Under Assumption~\ref{ass:dml}:
  \begin{enumerate}[label=(\roman*),leftmargin=*]
    \item the two estimators in a generic cell obey
          \begin{equation}\label{eq:joint-alr}
            \sqrt{n_c}
            \begin{pmatrix}
              \widehat\delta_L^{(X)}-\delta_L^{(X)} \\
              \widehat\theta_R^{(X)}-\theta_R^{(X)}
            \end{pmatrix}
            =
            \frac1{\sqrt{n_c}}\sum_{i:S_{c,i}=1}
            \begin{pmatrix}
              \phi_L^{P_c}(O_i) \\
              \phi_R^{(X),P_c}(O_i)
            \end{pmatrix}
            +o_{\Prb}(1);
          \end{equation}
    \item $\sqrt{n_c}(\widehat\delta_L^{(X)}-\delta_L^{(X)},\widehat\theta_R^{(X)}-\theta_R^{(X)})'\xrightarrow{d} N(0,\Sigma_c)$;
    \item $\widehat\Sigma_c=\mathbb P_{n,c}[\widehat{\boldsymbol\phi}_c^{P_c}\widehat{\boldsymbol\phi}_c^{P_c\prime}]\xrightarrow{p}\Sigma_c$;
    \item the estimators are regular, and $\Sigma_c$ equals the nonparametric efficiency bound for $(\delta_L^{(X)},\theta_R^{(X)})$.
  \end{enumerate}
\end{theorem}

\subsection{Proofs of adjusted population identities}

Conditional on $G=g$, total covariance gives
\[
  \Cov(D,Y_c\mid G=g)
  =
  \E[\Cov(D,Y_c\mid G=g,X)\mid G=g]
  +\Cov\{e_g(X),g_{g,t}(X)\mid G=g\}.
\]
The first term is $V_g^{(X)}\theta_{R,g,t}^{(X)}$.  Total variance gives
\[
  V_g
  =
  V_g^{(X)}+\Var\{e_g(X)\mid G=g\}.
\]
Their ratio proves \eqref{eq:response-0-X-relation}.

Under consistency and conditional PT-L,
\[
  m_{0,g,t}(X)=\E[U_c\mid G>t,X]
  =
  \E[U_c\mid G=g,X].
\]
Therefore
\[
  \E[Y_c-m_{0,g,t}(X)\mid G=g]
  =
  \E[\tau_c(D)\mid G=g]=\tau_{L,g,t},
\]
proving \eqref{eq:cohort-level-id-X}.

\begin{proof}[Proof of Proposition~\ref{prop:covariate-properties}]
  For every $\alpha\in\mathcal A_g^{(X)}$,
  \[
    1
    =
    \E[\alpha(D,X)\{D-e_g(X)\}\mid G=g].
  \]
  Cauchy--Schwarz implies
  \[
    \E[\alpha(D,X)^2\mid G=g]\ge\frac1{V_g^{(X)}}.
  \]
  The representer in \eqref{eq:alpha-X} belongs to the class and attains the bound.  Equality and normalization give uniqueness exactly as in the proof of Proposition~\ref{prop:minnorm}.
\end{proof}

Fix $x$ and apply the proof of Lemma~\ref{lem:weight-nonnegative} with conditional mean $e_g(x)$ and common positive scale $V_g^{(X)}$.  This proves $W_g^{(X)}(d,x)\ge0$.  Applying Lemma~\ref{lem:signed-connected} conditionally on $X=x$ to the function $q(d)=d$ gives
\[
  \int W_g^{(X)}(d,x)\dd d
  =
  \E[\alpha_g^{(X)}(D,X)D\mid G=g,X=x]
  =
  \frac{\Var(D\mid G=g,X=x)}{V_g^{(X)}}.
\]
Averaging proves \eqref{eq:cohort-W-X-global}.  Consistency, no anticipation, and conditional PT-R remove the additive conditional untreated trend.  Equation~\eqref{eq:diagonal-compatibility-X} and iterated expectations give, for almost every $x$,
\[
  \begin{aligned}
    \E[\alpha_g^{(X)}(D,x)Y_c\mid G=g,X=x]
    &=\E[\alpha_g^{(X)}(D,x)\tau_c(D)\mid G=g,X=x]\\
    &=\E[\alpha_g^{(X)}(D,x)q_{g,t,x}(D)\mid G=g,X=x].
  \end{aligned}
\]
A conditional application of the signed-measure lemma followed by the diagonal chain rule and averaging over $X\mid G=g$ proves \eqref{eq:cohort-selection-decomp-X}; the stipulated absolute integrability justifies the averaging.

\begin{proof}[Proof of Proposition~\ref{prop:adjusted-mixture}]
  Conditional on $X=x$, the zero dose for controls gives $\Var(D\mid X=x)=\pi V_D+\pi(1-\pi)\mu_D^2$ and $\Cov(D,Y\mid X=x)=\pi C_{DY}+\pi(1-\pi)\mu_D\Delta_Y$.  Integrating over $X$ gives denominator $A+B$ and numerator $A\theta_R^{(X)}+\E[\pi(X)\{1-\pi(X)\}\mu_D(X)\Delta_Y(X)]$, proving \eqref{eq:beta-decomp-raw}.  Consistency and conditional PT-L give the stated interpretation of $\Delta_Y(X)$.
\end{proof}

\subsection{Canonical gradients}

\begin{proof}[Proof of Proposition~\ref{prop:gradients}]
  Work under $P_c$ for a generic eligible cell and suppress $c$.  For the level functional, write
  \[
    \mu_1=\E[Y\mid T=1],
    \qquad
    \nu_0=\E[m_0(X)\mid T=1],
    \qquad
    \delta_L^{(X)}=\mu_1-\nu_0.
  \]
  The conditional-mean ratio gives
  \[
    IF_{\mu_1}(O)=\frac{T}{p}(Y-\mu_1).
  \]
  For $N_0=\E[Tm_0(X)]$, differentiation along a regular submodel with score $s(O)$ has the direct component $Tm_0(X)-N_0$ and the regression component
  \[
    \E[T\dot m_0(X)]
    =
    \E\left[
      (1-T)\frac{\pi(X)}{1-\pi(X)}
      \{Y-m_0(X)\}s(O)
      \right].
  \]
  The ratio rule for $\nu_0=N_0/p$ therefore gives
  \[
    IF_{\nu_0}(O)
    =
    \frac{T}{p}\{m_0(X)-\nu_0\}
    +\frac{1-T}{p}\frac{\pi(X)}{1-\pi(X)}
    \{Y-m_0(X)\}.
  \]
  Subtracting yields \eqref{eq:phiL}.

  For the response, first work under the treated-group distribution $\Q$.  Let
  \[
    N_R=\E_+[(D-e(X))(Y-g(X))],
    \qquad
    V^{(X)}=\E_+[(D-e(X))^2].
  \]
  The derivatives through $e$ and $g$ vanish because both residuals have conditional mean zero.  Thus influence functions for $N_R$ and $V^{(X)}$ are
  \[
    (D-e(X))(Y-g(X))-N_R
    \quad\text{and}\quad
    (D-e(X))^2-V^{(X)}.
  \]
  The ratio rule gives the gradient under $\Q$:
  \[
    \frac1{V^{(X)}}\left[
    (D-e(X))(Y-g(X))
    -\theta_R^{(X)}(D-e(X))^2
    \right].
  \]
  Transporting it from $\Q$ to the eligible-population distribution $P_c$ multiplies it by $T/p$, proving \eqref{eq:phiR}.  Both gradients are mean zero and lie in the unrestricted nonparametric tangent space, so they are canonical.
\end{proof}

\subsection{A cross-fitting lemma}

\begin{lemma}[Fixed-fold cross-fitting expansion]\label{lem:crossfit}
  Let $W_1,\ldots,W_n$ be i.i.d. and split into a fixed number $K$ of folds $I_k$.  Let $\widehat\eta_k$ be measurable with respect to observations outside $I_k$.  Suppose $\E[\psi(W;\eta_0)]=0$, $\E[\psi(W;\eta_0)^2]<\infty$, and, uniformly over folds,
  \[
    \norm{\psi(\cdot;\widehat\eta_k)-
      \psi(\cdot;\eta_0)}_{L_2(P)}=o_{\Prb}(1),
  \]
  \[
    \sqrt n\left|
    \E[\psi(W;\widehat\eta_k)\mid\widehat\eta_k]
    \right|=o_{\Prb}(1).
  \]
  Then
  \begin{equation}\label{eq:crossfit-expansion}
    \frac1{\sqrt n}\sum_{k=1}^K\sum_{i\in I_k}
    \psi(W_i;\widehat\eta_k)
    =
    \frac1{\sqrt n}\sum_{i=1}^n\psi(W_i;\eta_0)
    +o_{\Prb}(1).
  \end{equation}
  The result holds componentwise for a fixed-dimensional vector score.
\end{lemma}

\begin{proof}
  Let $\delta_k(W)=\psi(W;\widehat\eta_k)-\psi(W;\eta_0)$ and center it conditionally on the training data.  The conditional-bias assumption makes the sum of fold means $o_{\Prb}(1)$ after $\sqrt n$ scaling.  The centered contribution from fold $k$ has conditional variance at most $(|I_k|/n)\norm{\delta_k}_{L_2(P)}^2=o_{\Prb}(1)$.  Conditional Chebyshev's inequality and fixed $K$ prove \eqref{eq:crossfit-expansion}.
\end{proof}

\subsection{Exact nuisance remainders}

\begin{lemma}[Level-score remainder]\label{lem:level-rem}
  Let $\bar m=m_0+\delta_m$ and $\bar\pi=\pi+\delta_\pi$, with $\bar\pi$ bounded away from one.  Then
  \[
    \E[\psi_L(O;\delta_L^{(X)},\bar m,\bar\pi)]
    =
    \E\left[
      \delta_m(X)
      \frac{\bar\pi(X)-\pi(X)}{1-\bar\pi(X)}
      \right].
  \]
\end{lemma}

\begin{proof}
  Condition on $X=x$ and write $y_1(x)=\E[Y\mid T=1,X=x]$.  Since $\E[Y-\bar m(X)\mid T=0,X]=-\delta_m(X)$,
  \[
    \E[\psi_L\mid X]
    =
    \pi\{y_1-m_0-\delta_m-\delta_L^{(X)}\}
    +(1-\pi)\frac{\bar\pi}{1-\bar\pi}\delta_m.
  \]
  The expectation of $\pi(y_1-m_0-\delta_L^{(X)})$ is zero by definition of $\delta_L^{(X)}$.  The remaining coefficient on $\delta_m$ is $(\bar\pi-\pi)/(1-\bar\pi)$, proving the result.
\end{proof}

\begin{lemma}[Response-score remainder]\label{lem:response-rem}
  Let $\bar e=e+\delta_e$ and $\bar g=g+\delta_g$.  Then
  \[
    \E[\psi_R^{(X)}(O;\theta_R^{(X)},\bar e,\bar g)]
    =
    p\left\{
    \E_+[\delta_e(X)\delta_g(X)]
    -\theta_R^{(X)}\E_+[\delta_e(X)^2]
    \right\}.
  \]
\end{lemma}

\begin{proof}
  Under $\Q$,
  \[
    D-\bar e=D-e-\delta_e,
    \qquad
    Y-\bar g=Y-g-\delta_g.
  \]
  Expanding and using the two conditional mean-zero residuals gives
  \[
    \E_+[(D-\bar e)(Y-\bar g)]
    =
    N_R+\E_+[\delta_e\delta_g],
  \]
  \[
    \E_+[(D-\bar e)^2]
    =
    V^{(X)}+\E_+[\delta_e^2].
  \]
  Subtract $\theta_R^{(X)}$ times the second equality from the first and use $N_R=\theta_R^{(X)}V^{(X)}$.
\end{proof}

\subsection{DML asymptotic linearity}

\begin{proof}[Proof of Theorem~\ref{thm:joint-dml}]
  The level estimator solves the empirical level-score equation, whose derivative in $\delta$ is $-T$.  Hence
  \[
    \sqrt{n_c}(\widehat\delta_L^{(X)}-\delta_L^{(X)})
    =
    \frac{
      \sqrt{n_c}\mathbb P_{n,c}
      \psi_L(O;\delta_L^{(X)},\widehat m_{0,-k},\widehat\pi_{-k})
    }{\mathbb P_{n,c}[T]}.
  \]
  Lemma~\ref{lem:level-rem}, overlap, and Cauchy--Schwarz bound each fold-specific population bias by a constant times
  \[
    \norm{\widehat m_0-m_0}_2
    \norm{\widehat\pi-\pi}_2=o_{\Prb}(n^{-1/2}).
  \]
  Boundedness, propensity truncation, and nuisance consistency imply score $L_2$ convergence.  Lemma~\ref{lem:crossfit} and $\mathbb P_{n,c}[T]\xrightarrow{p}p$ give the first influence function expansion in \eqref{eq:joint-alr}.

  The response estimator likewise solves the empirical response-score equation:
  \[
    \sqrt{n_c}(\widehat\theta_R^{(X)}-\theta_R^{(X)})
    =
    \frac{
      \sqrt{n_c}\mathbb P_{n,c}
      \psi_R^{(X)}(O;\theta_R^{(X)},
      \widehat e_{-k},\widehat g_{-k})
    }{
      \mathbb P_{n,c}[T\{D-\widehat e_{-k}(X)\}^2]
    }.
  \]
  Lemma~\ref{lem:response-rem} bounds the fold-specific bias by
  \[
    C\left\{
    \norm{\widehat e-e}_2\norm{\widehat g-g}_2
    +\norm{\widehat e-e}_2^2
    \right\}=o_{\Prb}(n^{-1/2}).
  \]
  Boundedness and nuisance consistency give score $L_2$ convergence.  Moreover, conditional on a training fold,
  \[
    \E_+[\{D-\widehat e(X)\}^2\mid\widehat e]
    =
    V^{(X)}+\E_+[\{\widehat e(X)-e(X)\}^2\mid\widehat e],
  \]
  so the denominator converges to $pV^{(X)}$.  The cross-fitting lemma yields the second expansion.

  Stacking the two components and applying Cram\'er--Wold gives joint normality with covariance $\Sigma_c$, proving (ii).  The same boundedness and $L_2$ convergence arguments, together with the consistency of $\mathbb P_{n,c}[T]$, of the denominator in \eqref{eq:thetahat-X}, and of the two estimators, show that the plug-in gradient converges to $\boldsymbol\phi_c^{P_c}$ in $L_2$, so its empirical second moment is consistent, proving (iii).  For (iv), the influence function in \eqref{eq:joint-alr} is the canonical gradient of Proposition~\ref{prop:gradients} and the observed-data model is unrestricted, so the estimators are regular and their asymptotic covariance $\Sigma_c$ is the semiparametric efficiency bound, by the convolution theorem and the characterization of regular asymptotically linear estimators in \citet[Chapter~25]{vandervaart-1998}.
\end{proof}

\begin{proof}[Proof of Theorem~\ref{thm:staggered-dml}]
  For fixed $c$, let $s_c=P(S_c=1)$, which is positive by Assumption~\ref{ass:staggered-dml}, and $n_c=\sum_iS_{c,i}$.  Theorem~\ref{thm:joint-dml}(i) under $P_c$ gives
  \[
    \widehat{\boldsymbol\vartheta}_c^{(X)}
    -\boldsymbol\vartheta_c^{(X)}
    =
    \frac1{n_c}\sum_{i:S_{c,i}=1}
    \boldsymbol\phi_c^{(X),P_c}(O_i)
    +o_{\Prb}(n_c^{-1/2}).
  \]
  Since $n_c/n\xrightarrow{p}s_c$,
  \[
    \sqrt n
    (\widehat{\boldsymbol\vartheta}_c^{(X)}
    -\boldsymbol\vartheta_c^{(X)})
    =
    \frac1{\sqrt n}\sum_{i=1}^n
    \frac{S_{c,i}}{s_c}
    \boldsymbol\phi_c^{(X),P_c}(O_i)
    +o_{\Prb}(1),
  \]
  which is \eqref{eq:cell-full-if}.  Stack the influence functions over the fixed finite collection of cells.  Each influence function under $P_c$ is transported to the full-panel distribution $P$ by \eqref{eq:cell-full-if}, so every column is a function of the same observation $O_i$, and the stacked vector is i.i.d., mean zero, and square integrable.  Cram\'er--Wold yields \eqref{eq:staggered-vector-clt}.  For \eqref{eq:staggered-sigma}, $\mathbb P_n[S_c]\xrightarrow{p}s_c$ and the $L_2$ convergence of each cell's plug-in gradient in the proof of Theorem~\ref{thm:joint-dml}(iii) give $\mathbb P_n[\norm{\widehat{\boldsymbol\phi}_i^{(X),P}-\boldsymbol\phi_i^{(X),P}}^2]=o_{\Prb}(1)$; Cauchy--Schwarz entry by entry and the law of large numbers then prove covariance consistency.  Theorem~\ref{thm:event-aggregation} then gives event-time inference for known or regularly estimated treatment-timing weights, provided the estimated aggregate influence functions satisfy its empirical-$L_2$ condition.  The cell-level nuisance-rate conditions are unchanged.
\end{proof}

\clearpage
\section{Additional application results}\label{app:fracking-checks}

\subsection{CGBS samples and control groups}

The CGBS Figure~2 code regresses treated counties' employment changes, net of estimated zero-dose control means, on a cubic spline in dose. The short window (event times 0--2) uses 329 treated counties in eight cohorts; the long window (3--4) uses 307 counties in seven cohorts. The archive also contains an overall window (0--4) for those seven cohorts. Section~\ref{sec:application} uses seven cohorts in every window and date-specific $G>t$ controls.

In the archived spline code, the control sample consists of strict zero-dose counties. There are 73 in the short window and 60 in the long and overall windows. For the latter windows, the code additionally requires a control county's recorded adoption year plus the window's last event time to fall within the panel. This restriction removes 13 zero-dose counties even though the preparation otherwise codes them as never treated.

To isolate the choice of controls within our specification, we hold the cohorts, outcome differences, and aggregation fixed and replace $G>t$ controls with all 73 zero-dose counties. The short-window level contrast falls from 0.0094 to 0.0065, a difference of 0.0030 (i.i.d. SE 0.0024; play SE 0.0049). The long-window contrast falls from 0.0432 to 0.0388, a difference of 0.0044 (i.i.d. SE 0.0027; play SE 0.0064). Both differences are imprecisely estimated. This comparison changes only the control group in our level estimator and does not reproduce the CGBS spline estimand.

\subsection{Inference for the CGBS benchmark}

The 14 plays in the 402-county sample have an effective count of $(\sum_h n_h)^2/\sum_h n_h^2=6.97$, where $n_h$ is the number of sample counties in play $h$. This measures concentration in cluster sizes.

Table~\ref{tab:fracking-cgbs-inference} reports the CGBS level estimates at dose four. The archived influence function standard errors account only for the treated observations and treat the estimated zero-dose control means as fixed. Adding the influence of those means raises the long-window standard error from 0.0091 to 0.0131 and the overall-window standard error from 0.0060 to 0.0089. This correction holds the empirical spline knots and boundaries fixed; it does not account for selection of the spline basis.

\begin{table}[!htbp]
  \centering
  \caption{Inference for the CGBS level estimates at dose four}
  \label{tab:fracking-cgbs-inference}
  {\renewcommand{\arraystretch}{1.15}
    \begin{tabular}{lrrrr}
\toprule
Window & Estimate & Original SE & iid SE & Play SE \\ 
\midrule
Short & 0.0059 & 0.0041 & 0.0061 & 0.0077 \\
Long & 0.0277 & 0.0091 & 0.0131 & 0.0226 \\
Overall & 0.0163 & 0.0060 & 0.0089 & 0.0138 \\
\bottomrule
\end{tabular}
}
  \vspace{0.5em}
  \begin{minipage}{0.98\textwidth}
    \begin{singlespace}
      \footnotesize
      \textit{Notes:} Estimates are in log points. Short averages event times 0--2; long averages 3--4; overall averages 0--4. Original SE reproduces the treated-only influence function standard error in the CGBS archive. The i.i.d. SE adds uncertainty from the estimated zero-dose control means conditional on the empirical spline basis. Play SE applies CR1 clustering by shale play to the corrected influence functions. The clustered standard errors provide a sensitivity analysis for within-play dependence; the paper's sampling results assume independent units.
    \end{singlespace}
  \end{minipage}
\end{table}

The archived long- and overall-window intervals exclude zero. After the control-mean correction, the i.i.d. long-window interval still excludes zero but the overall interval does not. The corresponding play-clustered standard errors are 0.0226 and 0.0138, and both intervals include zero. These are pointwise comparisons. The published CGBS figures also report pointwise intervals, whereas Figure~\ref{fig:fracking-margins} uses simultaneous intervals for the fixed collection of ten estimates under county-level i.i.d. sampling.

\subsection{Missing scores and the aggregate proportionality restriction}

Excluding the 44 counties whose missing prospectivity scores were coded as zero leaves 358 counties: 329 with positive scores and 29 with recorded zeros. With the same cells and aggregation rules, the short-, long-, and overall-window level contrasts change from 0.0094, 0.0432, and 0.0229 to 0.0138, 0.0513, and 0.0288. The response estimates remain 0.0088, 0.0272, and 0.0162 because the omitted counties enter only the controls.

For each event time $e=0,\ldots,4$, we aggregate $\delta_{L,g,g+e}-\mu_{D,g}\theta_{R,g,g+e}$ with full-panel cohort shares. Inference includes estimation of $\mu_{D,g}$ through \eqref{eq:phi-proportionality} and of the shares. The joint zero restriction on the five aggregates has an i.i.d. Wald $p$-value of 0.228 and a descriptive play-CR1 Wald $p$-value of 0.337. Nonrejection does not establish equality in all 35 cells, and the restriction is distinct from PT-L and PT-R.

\clearpage
\section{Illustrations of the two margins}\label{app:illustrations}

\begingroup
\setlength{\intextsep}{8pt}
\setlength{\floatsep}{8pt}
\begin{figure}[!htb]
  \centering
  \caption{Separate roles of PT-L and PT-R}
  \label{fig:pt-roles}
  \vspace{0.35em}
  \includegraphics[width=0.96\textwidth]{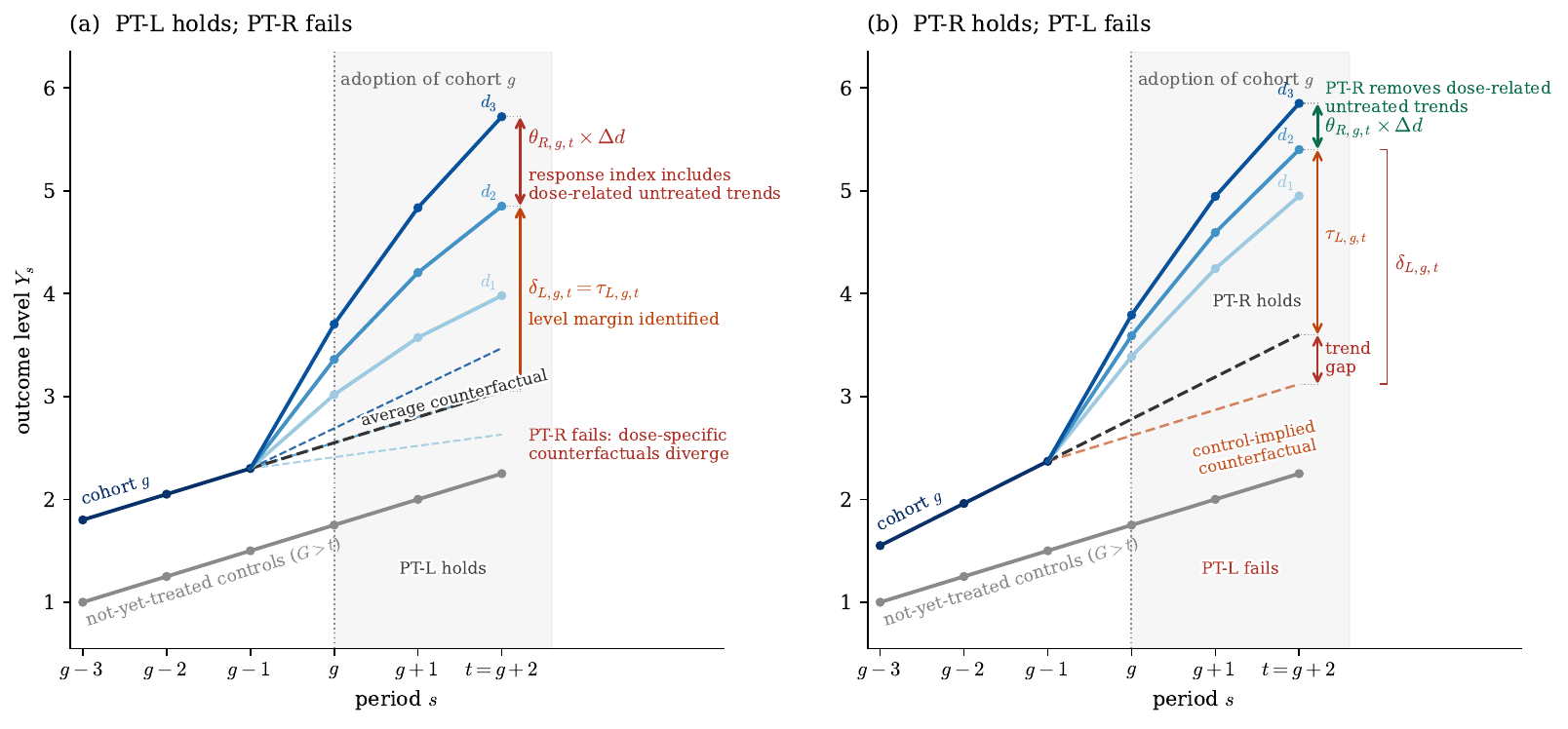}
  \vspace{0.35em}
  \begin{minipage}{0.96\textwidth}
    \begin{singlespace}
      \footnotesize
      \textit{Notes:} Stylized linear example with three equally weighted doses and $d_2=\mu_{D,g}$. PT-L identifies the level margin; PT-R removes dose-related untreated trends from the response index.
    \end{singlespace}
  \end{minipage}
\end{figure}

\begin{figure}[!htb]
  \centering
  \caption{The two margins and the cellwise continuous-dose OLS coefficient}
  \label{fig:two-margins-cell}
  \vspace{0.35em}
  \includegraphics[width=0.96\textwidth]{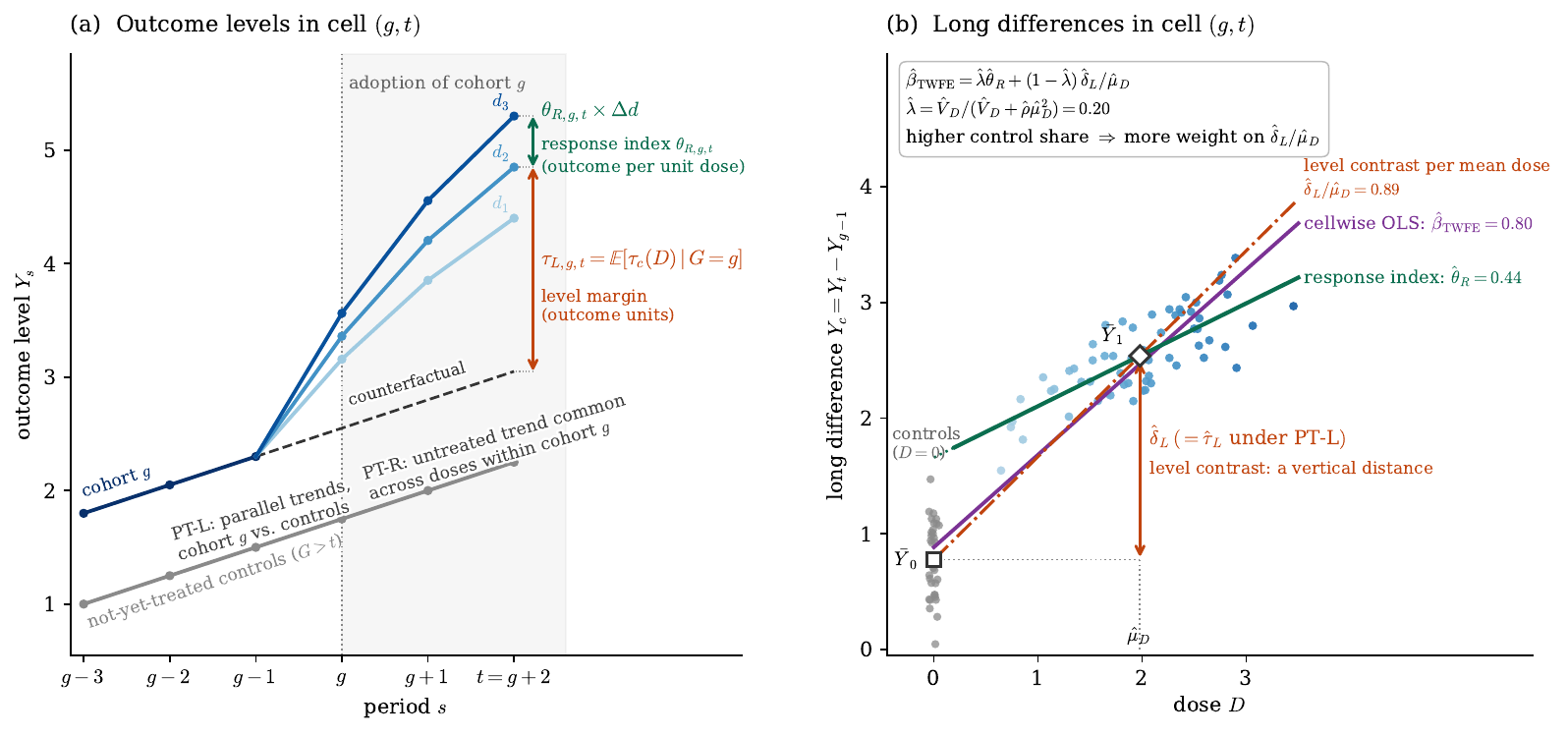}
  \vspace{0.35em}
  \begin{minipage}{0.96\textwidth}
    \begin{singlespace}
      \footnotesize
      \textit{Notes:} Stylized linear example. Panel~A uses three equally weighted doses, with $d_2=\mu_{D,g}$. Panel~B uses simulated data to illustrate Corollary~\ref{cor:share} for a valid cohort-time cell.
    \end{singlespace}
  \end{minipage}
\end{figure}
\endgroup

% End of the self-contained article.
\end{document}